\documentclass{IEEEtran}
\usepackage{cite}
\usepackage{amsmath,amssymb,amsfonts}
\usepackage{subfigure}
\usepackage{algorithmic}
\usepackage{graphicx}
\usepackage{textcomp}
\usepackage{caption}
\usepackage{color}
\usepackage{hhline}
\usepackage{tabu,multirow}
\usepackage{makecell}
\usepackage{epstopdf}
\usepackage{geometry}
\usepackage{indentfirst}
\usepackage[T1]{fontenc}
\usepackage{float}
\usepackage{subeqnarray}
\usepackage{cases}
\usepackage{times}
\usepackage{booktabs}
\usepackage{multirow}
\usepackage{siunitx}

\input{tcilatex}
\begin{document}

\title{Anti-Unwinding Sliding Mode Attitude Maneuver Control for Rigid
Spacecraft}
\author{Rui-Qi Dong, \IEEEmembership{Student Member, IEEE}, Ai-Guo Wu$^\ast$%
, \IEEEmembership{Member, IEEE} and Ying Zhang
\thanks{%
This work was supported in part by the National Natural Science Foundation
of China for Excellent Young Scholars under Grant No. 61822305; Major
Program of National Natural Science Foundation of China under Grant Numbers
61690210 and 61690212; The Fundamental Research Funds for the Central
Universities under Grant No. HIT.BRETIV.201907; Guangdong Natural Science
Foundation under Grant No. 2017A030313340; Shenzhen Municipal Project for
International Cooperation with Project No. GJHZ20180420180849805; Shenzhen
Municipal Basic Research Project for Discipline Layout with Project No.
JCYJ20180507183437860. \newline
$^\ast$ Corresponding author.} \thanks{%
The authors are with the Department of Mechanical Engineering and
Automation, Harbin Institute of Technology (Shenzhen), Shenzhen 518055
(e-mail: rykidong@163.com; ag.wu@163.com; zhangyinghit@126.com). } }
\maketitle

\begin{abstract}
In this paper, anti-unwinding attitude maneuver control for rigid spacecraft
is considered. First, in order to avoid the unwinding phenomenon when the
system states are restricted to the switching surface, a novel switching
function is constructed by hyperbolic sine functions such that the switching
surface contains two equilibriums. Then, a sliding mode attitude maneuver
controller is designed based on the constructed switching function to ensure
the robustness of the closed-loop attitude maneuver control system to
disturbance. Another important feature of the developed attitude control law
is that a dynamic parameter is introduced to guarantee the anti-unwinding
performance before the system states reach the switching surface. The
simulation results demonstrate that the unwinding problem is settled during
attitude maneuver for rigid spacecraft by adopting the newly constructed
switching function and proposed attitude control scheme.
\end{abstract}


\section{Introduction}

Due to the increasingly challenging requirements of the aerospace control
tasks such as high pointing accuracy, fast response and strong robustness,
the attitude controller design for a spacecraft has been a hot topic. Then,
various control schemes have been proposed to deal with the attitude control
issue, such as Proportional-Integral-Differential (PID) control law~\cite%
{li2016robust}, Linear Parameter Varying (LPV) gain-scheduled controller~%
\cite{jin2018lpv}, fuzzy control method~\cite{chen2000adaptive},
velocity-free approach~\cite{guo2019velocity}, robust $H_{\infty }$ control
technique~\cite{gao2012robust}, and so on. Despite all of these efforts, the
attitude maneuver control of rigid spacecraft is still challenging.
%This is mainly due to three factors, i.e., nonlinearity, disturbance, and unwinding phenomenon.

Sliding mode control (SMC) is a nonlinear control technique that alters the
dynamics of a nonlinear system by application of a discontinuous control law
(or more rigorously, a set-valued control signal) that forces the system to
"slide" along a cross-section of the system's normal behavior. Such a
control technique was first proposed in~\cite{utkin1977variable} for
variable structure systems. Subsequently, it has attracted much attention in
handling spacecraft attitude control design because of its strong robustness~%
\cite{dodds1991sliding,bang2005flexible,hu2006control,bang2005flexible}. In~%
\cite{dodds1991sliding}, an SMC scheme was developed for a three-axis
attitude control of rigid spacecraft with unknown dynamic parameters. Then,
the SMC control strategy for the pure rigid spacecraft was extended to a
flexible spacecraft, and an SMC strategy for the flexible spacecraft
attitude maneuver was proposed in~\cite{bang2005flexible}. In~\cite%
{hu2006control}, an SMC output feedback control law was presented to solve
the attitude stabilization problem for the flexible spacecraft with
uncertainty, disturbances, and control input nonlinearities. In~\cite%
{hu2006adaptive}, an SMC controller was derived for the attitude maneuver
problem of a flexible spacecraft under control input nonlinearities, and
only the attitude and angular rate information were used. Due to the
existence of the sign function, the traditional SMC controllers suffer from
the chattering problem. In order to alleviate such undesirable performance,
the sign function was approximated by a saturation function~\cite%
{chen1993sliding,zong2010higher}.
%The system states were no longer forced to stay in
%the switching surface but were constrained within the sliding layer. A
%bipolar sigmoid function on-line adaptation was developed in~%
%\citeA{zong2010higher} to eliminate chattering, and an adjustable control
%gain tuning approach without high-frequency switching was proposed to
%estimate the unknown upper bound of the system uncertainties.
In~\cite{tiwari2016attitude}, a higher-order integral SMC method was
presented for the attitude control of rigid spacecraft, which was free of
chattering because the nonlinear term was introduced into the first
derivative of the control input.
%The cost of such substitution was a reduction in the accuracy of the desired performance.
It should be pointed out that the aforementioned control methods were
developed based on a linear sliding surface, and thus the system states
reach their equilibrium point in infinite time rather than finite time~\cite%
{li2006global}. Recently, a terminal sliding mode control methodology has
been proposed, in which a nonlinear sliding surface was synthesized to
achieve the finite-time control performance~\cite%
{wu2011quaternion,wu2018adaptive}. In~\cite{wu2011quaternion}, the developed
finite-time SMC law can guarantee the convergence of attitude tracking
errors in finite time. % rather than in the asymptotic sense.
In~\cite{wu2018adaptive}, by constructing a nonlinear sliding surface,
adaptive finite-time SMC algorithms were presented to stabilize the flexible
spacecraft attitude. % in the presence of the inertia uncertainty.
In addition, the SMC technique was also combined with other control methods
to obtain enhanced control performance for spacecraft attitude control, such
as backstepping method~\cite{ji2018vibration}, adaptive control~\cite%
{guo2014adaptive}.

A typical feature in most of the control approaches mentioned above for
spacecraft is that the unwinding issue was ignored when the spacecraft
attitude is described by quaternions. The quaternion has a double value
property, and thus there are two mathematical representations for a given
physical attitude of a rigid body~\cite{hu2015spacecraft}. Accordingly,
there are two equilibriums $\left[1,\ 0,\ 0,\ 0\right]^{\mathrm{T}}$ and $%
\left[-1,\ 0,\ 0,\ 0\right]^{\mathrm{T}}$. However, in conventional control
law design, only one equilibrium is considered.
%In this case, if a spacecraft system is stabilized but $-1$ is taken to describe the spacecraft attitude, then any disturbance can lead the spacecraft to rotate $360$
In this case, the system states have to move to the considered equilibrium,
even if they are very close to another equilibrium.
%In this case, a spacecraft must rotate $360^{\circ}$ to change its attitude coordinates from $-1$ to $1$. the unstable may makes a spacecraft
%perform an unnecessary large-angle maneuver when a small-angle maneuver in
%the opposite rotational direction is sufficient to achieve the objective
%even if they are very close to another equilibrium.
This is called the unwinding phenomenon, which may cause a spacecraft to
perform an unnecessary large-angle maneuver when a small-angle maneuver %in
%the opposite rotational direction
is sufficient to achieve the control objective. To the best knowledge of the
author, there are little research about the unwinding issue for spacecraft.
%New attitude error
%functions were proposed in \cite{unwinding2} and \cite{unwinding3}, respectively.
In~\cite{tiwari2018spacecraft}, the term $\mathrm{sign}\left(q_{0}(0)\right)
$ was introduced into the sliding surface to avoid unwinding phenomenon. In
\cite{kristiansen2005satellite}, a new attitude error function $%
\left(1-|q_{0}|\right)$ was constructed to design attitude controller, which
considers two equilibrium $\left[1,\ 0,\ 0,\ 0\right]^{\mathrm{T}}$ and $%
\left[-1,\ 0,\ 0,\ 0\right]^{\mathrm{T}}$. But the strict proof of how the
designed control laws avoid unwinding phenomenon was not given.

In this paper, the unwinding phenomenon is taken into account, and an
anti-unwinding sliding mode attitude maneuver control law for rigid
spacecraft is presented. The main contribution of this work can be
summarized as follows. First of all, a novel switching function that
contains two equilibrium points is developed. Moreover, the anti-unwinding
performance is proven when the system states are on the switching surface by
constructing a Lyapunov function. This Lyapunov function is constructed by a
hyperbolic cosine function. Secondly, a sliding mode control law is designed
to guarantee that all the system trajectories are attracted by the switching
surface. Further, the anti-unwinding performance is proven by designing a
dynamic parameter for the sliding mode control law.
%Thirdly, the rigorous theoretical proof of anti-unwinding performance when system states are on the switching surface and outside the switching surface are given, respectively.

This paper proceeds as follows. In Section II, an attitude maneuver control
problem of rigid spacecraft is stated. In Section III, a novel switching
function is first constructed, and the property of the switching surface is
analyzed. Furthermore, an anti-unwinding sliding mode controller is
presented, and its anti-unwinding performance is proven. In Section V,
comparing simulations are conducted to demonstrate the efficiency of the
proposed attitude maneuver controller.%kinematics and

Throughout this paper, we use the italic-font notation for a scalar variable
(as $\alpha$), the bold-font notation for a vector (as $\boldsymbol{v}$),
and the capital-letter notation for a matrix (as $M$). The set of $n$%
-dimensional real vectors and the set of $m$-by-$n$ real matrices are
denoted by $\mathbb{R}^{n}$ and $\mathbb{R}^{m\times n}$, respectively. We
use $\left\Vert \cdot \right\Vert $ to represent the $2$-norm of a vector, $%
\lambda _{\min }\left( \cdot\right)$ and $\lambda _{\max }\left( \cdot\right)
$ to represent the minimum and maximum eigenvalues of a matrix,
respectively. %,or the Frobenius norm of a matrix.
%In addition, %we use it is assumed that functions $\phi\left(x\right)$ and $y=f\left(u\right)$ are derivable at $x$ and $u=\phi\left(x\right)$, respectively. Then,
%we use $y^{\prime }\left(u\right)$ and $\dot{y}\left(x\right)$ to denote the derivation of a compound function $y\left(u\left(x\right)\right)$ at point $u=\phi\left(x\right)$ and $x$, respectively, such that $y^{\prime }\left(u\right)=\frac{\mathrm{d}y}{\mathrm{d}u}$, and $\dot{y}\left(x\right)=\frac{\mathrm{d}y}{\mathrm{d}u}\cdot\frac{\mathrm{d}u}{\mathrm{d}x}$ . Besides,
In addition, the following two hyperbolic functions are used, $\sinh x=\frac{%
e^{x}-e^{-x}}{2},\ \cosh x=\frac{e^{x}+e^{-x}}{2}.$ %\[
%\sinh=\frac{e^{x}-e^{-x}}{2},\ \cosh=\frac{e^{x}+e^{-x}}{2}.
%\]
Moreover, the following derivatives are used, $\frac{\mathrm{d}\left(\sinh
x\right)}{\mathrm{d}x}=\cosh x,\ \frac{\mathrm{d}\left(\cosh x\right)}{%
\mathrm{d}x}=\sinh x,\ \frac{\mathrm{d}\left( \arccos x\right) }{\mathrm{d}t}%
=-\frac{x}{\sqrt{1-{x}^{2}}},\ x\in \mathbb{R}$, respectively.

\section{Attitude Maneuver Control Problem Formulation for a Rigid Spacecraft%
}

%\subsection{Attitude maneuver problem statement}
%In this work, we tackle a rest-to-rest attitude maneuver control problem for
%a rigid spacecraft. In order to construct a mathematical model for this problem, three frames are used, which are shown in Fig.~\ref{fig1}, where the inertia
%frame $\mathcal{F}_{\mathrm{I}}$ and desired frame $\mathcal{F}_{\mathrm{d}}$
%are fixed frames, the body frame $\mathcal{F}_{\mathrm{b}}$ is time-varying. The spacecraft attitude maneuver task is to rotate
%an Euler
%angle $\theta \in \left(0,2\pi\right) $ around the Euler axis $\boldsymbol{e}$
%the body frame $\mathcal{F}_{\mathrm{b}}$ to the desired attitude frame $%
%\mathcal{F}_{\mathrm{d}}$. In addition, according to the Euler theorem, the
%attitude maneuver for a rigid spacecraft can be further described as to
%rotate a rigid spacecraft an Euler angle $\theta \in \left(0,2\pi\right) $
%around the Euler axis $\boldsymbol{e}=\left[e_1,\ e_2,\ e_3\right]$ from $\mathcal{F}_{\mathrm{b}}$ to $%
%\mathcal{F}_{\mathrm{d}}$.

In this paper, we aim to design an attitude maneuver controller to rotate
the rigid spacecraft from the body frame $\mathcal{F}_{\mathrm{b}}$ to the
desired frame $\mathcal{F}_{\mathrm{d}}$. For this end, the attitude
dynamics of the body frame $\mathcal{F}_{\mathrm{b}}$ is given in the
subsequent section. %In the following, the challenges of designing an

\subsection{Rigid Spacecraft Attitude Kinematics and Dynamics}

The quaternion based kinematic and dynamic equations of a rigid spacecraft
can be given by~\cite{chen1993sliding} %\\cite{hughes2012spacecraft},
\begin{equation}
\left\{
\begin{tabular}{l}
$\boldsymbol{\dot{q}}=\dfrac{1}{2}\left[
\begin{array}{c}
-\boldsymbol{q}_{\mathrm{v}}^{\mathrm{T}} \\
q_{0}I_{3}+\boldsymbol{q}_{\mathrm{v}}^{\times }%
\end{array}%
\right] \boldsymbol{\omega },$ \\
\makecell[{ll}]{$J \boldsymbol{\dot{\omega}}=-\boldsymbol{\omega }^{\times }
J\boldsymbol{\omega }+\boldsymbol{u}+\boldsymbol{d},$}%
\end{tabular}%
\right.  \label{dynamic equaiton1}
\end{equation}%
where unit quaternion $\boldsymbol{q}=\left[ q_{0}\ \boldsymbol{q}_{\mathrm{v%
}}^{\mathrm{T}} \right] ^{\mathrm{T}}\in \mathbb{R}\times \mathbb{R}^{3}$
represents the attitude of body frame $\mathcal{F}_{\mathrm{b}}$ with
respect to inertia frame $\mathcal{F}_{\mathrm{I}}$, $\boldsymbol{\omega }%
\in \mathbb{R}^{3}$ denotes the angular velocity of body frame $\mathcal{F}_{%
\mathrm{b}}$ with respect to inertia frame $\mathcal{F}_{\mathrm{I}}$;
% expressed in body frame $\mathcal{F}_{\mathrm{b}}$;
$J\in \mathbb{R}^{3\times 3}$ is the inertia matrix (symmetric) of the whole
rigid spacecraft, $\boldsymbol{u}$ is the external torque acting on the main
body, and $\boldsymbol{d}$ is the external disturbance. In addition, for any
vector $\boldsymbol{x}\in \mathbb{R}^{3}$, $\boldsymbol{x}^{\times }$
represents a skew-symmetric matrix which can be given by
\begin{equation*}
\boldsymbol{x}^{\times }:=\left[
\begin{array}{ccc}
0 & -x_{3} & x_{2} \\
x_{3} & 0 & -x_{1} \\
-x_{2} & x_{1} & 0%
\end{array}%
\right] .
\end{equation*}

%In order to design an attitude controller to maneuver the rigid spacecraft
%from body frame $\mathcal{F}_{\mathrm{b}}$ to desired frame $\mathcal{F}_{\mathrm{d}}$, the error
%kinematics and dynamics between $\mathcal{F}_{\mathrm{b}}$ and $\mathcal{F}_{%
%\mathrm{d}}$ are given in the subsequent section.

%The spacecraft attitude maneuver task is to rotate an Euler angle $\theta \in \left(0,2\pi\right) $ around the Euler axis $\boldsymbol{e}$ the body frame $\mathcal{F}_{\mathrm{b}}$ to the desired attitude frame $\mathcal{F}_{\mathrm{d}}$.

Based on the attitude dynamics of the body frame $\mathcal{F}_{\mathrm{b}}$,
the error kinematics and dynamics between the body frame $\mathcal{F}_{%
\mathrm{b}}$ and the desired frame $\mathcal{F}_{\mathrm{d}}$ are given in
the next section.

\subsection{Relative Attitude Error Kinematics and Dynamics}

\subsubsection{Attitude Error Kinematics}

Let unit quaternion $\boldsymbol{q}_{\mathrm{d}}:=\left[ q_{\mathrm{d}0}\
\boldsymbol{q}_{\mathrm{dv}}^{\mathrm{T}}\right] ^{\mathrm{T}}\in \mathbb{R}%
\times \mathbb{R}^{3}$ represents the rigid spacecraft attitude of desired
frame $\mathcal{F}_{\mathrm{d}}$ with respect to inertia frame $\mathcal{F}_{%
\mathrm{I}}$. Let $\boldsymbol{\omega }_{\mathrm{d}}\in \mathbb{R}^{3}$
denotes the rigid spacecraft angular velocity of $\mathcal{F}_{\mathrm{b}}$
with respect to $\mathcal{F}_{\mathrm{I}}$ and is expressed in $\mathcal{F}_{%
\mathrm{b}}$. The attitude error $\boldsymbol{q}_{\mathrm{e}}:=\left[ q_{%
\mathrm{e}0}\ \boldsymbol{q}_{\mathrm{ev}}^{\mathrm{T}}\right] ^{\mathrm{T}%
}\in \mathbb{R}\times \mathbb{R}^{3}$ can be given by
\begin{equation}
\boldsymbol{q}_{\mathrm{e}}=\boldsymbol{q}_{\mathrm{d}}^{\boldsymbol{\ast }%
}\otimes \boldsymbol{q},  \label{errorq}
\end{equation}%
where $\boldsymbol{q}_{\mathrm{d}}^{\boldsymbol{\ast }}:=\left[ \boldsymbol{q%
}_{\mathrm{d}0}\ -\boldsymbol{q}_{\mathrm{dv}}^{\mathrm{T}}\right] ^{\mathrm{%
T}}$,
%$\boldsymbol{q}_{\mathrm{d}}^{\boldsymbol{\ast }}$ is the conjugate of $\boldsymbol{q}_{\mathrm{d}}$ defined by $\boldsymbol{q}_{\mathrm{d}}^{\boldsymbol{\ast }}:=\left[ \boldsymbol{q}_{\mathrm{d}0}\ -\boldsymbol{q}_{\mathrm{dv}}^{\mathrm{T}}\right] ^{\mathrm{T}}$,
and $\otimes $ is the quaternion multiplication operator. %, such that
%\begin{equation*}
%\boldsymbol{q}_{\mathrm{d}}^{\boldsymbol{\ast }}\otimes \boldsymbol{q}=\left[
%\begin{array}{cccc}
%q_{\mathrm{d0}} & q_{\mathrm{d1}} & q_{\mathrm{d2}} & q_{\mathrm{d3}} \\
%-q_{\mathrm{d1}} & q_{\mathrm{d0}} & q_{\mathrm{d3}} & -q_{\mathrm{d2}} \\
%-q_{\mathrm{d2}} & -q_{\mathrm{d3}} & q_{\mathrm{d0}} & q_{\mathrm{d1}} \\
%-q_{\mathrm{d3}} & q_{\mathrm{d2}} & -q_{\mathrm{d1}} & q_{\mathrm{d0}}%
%\end{array}%
%\right] \left[
%\begin{array}{c}
%q_{\mathrm{0}} \\
%q_{\mathrm{1}} \\
%q_{\mathrm{2}} \\
%q_{\mathrm{3}}%
%\end{array}%
%\right] .
%\end{equation*}%
Then, the components of the error quaternion $\boldsymbol{q}_{\text{\textrm{e%
}}}$ can be obtained from (\ref{errorq}),
\begin{align}
q_{\mathrm{e}0}=& \boldsymbol{q}_{\mathrm{dv}}^{\mathrm{T}}\boldsymbol{q}_{%
\mathrm{v}}+q_{\mathrm{d}0}q_{0},  \label{multi} \\
\boldsymbol{q}_{\mathrm{ev}}=& q_{\mathrm{d}0}\boldsymbol{q}_{\mathrm{v}}-%
\boldsymbol{q}_{\mathrm{dv}}^{\mathrm{\times }}\boldsymbol{q}_{\mathrm{v}%
}-q_{0}\boldsymbol{q}_{\mathrm{dv}}.  \notag
\end{align}%
Moreover, it can be derived from (\ref{multi}) that
\begin{equation}
q_{\mathrm{e}0}^{2}+\boldsymbol{q}_{\mathrm{ev}}^{\mathrm{T}}\boldsymbol{q}_{%
\mathrm{ev}}=1.  \label{unitquaternion}
\end{equation}%
By taking derivative for~(\ref{multi}), the following attitude error
kinematics % in terms of error quaternion $\boldsymbol{q}_{\mathrm{e}}$
can be obtained as
\begin{equation}
\dot{\boldsymbol{q}}_{\mathrm{e}}=\dfrac{1}{2}\left[
\begin{array}{c}
-\boldsymbol{q}_{\mathrm{ev}}^{\mathrm{T}} \\
q_{\mathrm{e}0}I_{3}+\boldsymbol{q}_{\mathrm{ev}}^{\times }%
\end{array}%
\right] \boldsymbol{\omega }_{\mathrm{e}},  \label{kinematic1}
\end{equation}%
where $\boldsymbol{\omega }_{\mathrm{e}}\in \mathbb{R}^{3}$ represents the
angular velocity error, and is defined as%
\begin{equation}
\boldsymbol{\omega }_{\mathrm{e}}:=\boldsymbol{\omega }-R\boldsymbol{\omega }%
_{\mathrm{d}},  \label{transformation1}
\end{equation}%
with $R$ being the relative rotation matrix from $\mathcal{F}_{\mathrm{b}}$
to $\mathcal{F}_{\mathrm{d}}$, which is given by
\begin{equation*}
R:=\left( q_{\mathrm{e}0}^{2}-\boldsymbol{q}_{\mathrm{ev}}^{\mathrm{T}}%
\boldsymbol{q}_{\mathrm{ev}}\right) I_{3}+2\boldsymbol{q}_{\mathrm{ev}}%
\boldsymbol{q}_{\mathrm{ev}}^{\mathrm{T}}-2q_{\mathrm{e}0}\boldsymbol{q}_{%
\mathrm{ev}}^{\times }.
\end{equation*}%
The rotation matrix $R$ satisfies %$\left\Vert R\right\Vert =1$ and
$\dot{R}=-\boldsymbol{\omega }_{\mathrm{e}}^{\times }R$. Furthermore, it can
be obtained from~(\ref{transformation1}) that
\begin{equation}
\dot{\boldsymbol{\omega}}_{\mathrm{e}}=\dot{\boldsymbol{\omega}}+\boldsymbol{%
\omega }_{\mathrm{e}}^{\times }R\boldsymbol{\omega }_{\mathrm{d}}-R\dot{%
\boldsymbol{\omega}}_{\mathrm{d}}.  \label{intermediate}
\end{equation}

\subsubsection{Attitude Error Dynamics}

For a rest-to-rest attitude maneuver control problem, the desired attitude
velocity satisfies $\boldsymbol{\omega }_{\mathrm{d}}=0,\ \dot{\boldsymbol{%
\omega}}_{\mathrm{d}}=0$. Thus, it can be obtained from~(\ref%
{transformation1}) that~$\boldsymbol{\omega }_{\mathrm{e}}=\boldsymbol{%
\omega }$ holds. With this in mind,
%In this work, a attitude rest-to-rest attitude maneuver issue is considered, that means $\omega _{\mathrm{d}}=0,\dot{\omega}_{\mathrm{d}}=0$
by substituting~(\ref{transformation1}) and~(\ref{intermediate}) into the
second equation of~(\ref{dynamic equaiton1}), the following attitude error
dynamic equation can be obtained,%
\begin{equation}
J\dot{\boldsymbol{\omega}}_{\mathrm{e}}=-\boldsymbol{\omega }_{\mathrm{e}%
}^{\times }J\boldsymbol{\omega }_{\mathrm{e}}+\boldsymbol{u}+\boldsymbol{d}.
\label{dynamicequation2}
\end{equation}
%\subsection{Inertia uncertainty and disturbance}
%\quad Furthermore, the disturbance and inertia uncertainty are considered in
%this paper. %Then, a rigid spacecraft attitude control system model in the presence of disturbance and inertia uncertainty is further constructed.
%It is assumed that the nominal inertia matrix of a rigid spacecraft is $J_{0}$, and the uncertainty inertia matrix is $\Delta J$. Then, the inertia matrix $J$ of the whole spacecraft can be obtained as,
%\begin{equation}
%J:=J_{\mathrm{0}}+\Delta J,  \label{Jinertia}
%\end{equation}%
%Define the inertia matrix $J$ of the whole spacecraft as
%\begin{equation}
%J:=J_{\mathrm{0}}+\Delta J,  \label{Jinertia}
%\end{equation}%
%where $\Delta J$ is the inertia uncertainty caused by the change of mass properties,
%and $J_{\mathrm{0}}$ is a constant matrix representing the nominal inertia matrix.
%In addition, a new disturbance signal $\boldsymbol{\bar{d}}$ is
%introduced as a combination of the original disturbance signal $\boldsymbol{d%
%}$, and the signal $\boldsymbol{d}_{\mathrm{un}}$ associated with inertia
%uncertainty such that%
%\begin{equation}
%\boldsymbol{\bar{d}}:=\boldsymbol{d}+\boldsymbol{d}_{\mathrm{un}},
%\label{bar_d}
%\end{equation}%
%with
%\begin{equation*}
%\boldsymbol{d}_{\mathrm{un}}:=-\Delta J\boldsymbol{\dot{\omega}}_{\mathrm{e}%
%}-\boldsymbol{\omega }_{\mathrm{e}}^{\times }\Delta J\boldsymbol{\omega }_{%
%\mathrm{e}}.
%\end{equation*}%
%Then, by (\ref{kinematic1}), (\ref{Jinertia}), and (\ref{bar_d}),
Then, by (\ref{kinematic1}) and (\ref{dynamicequation2}), the attitude error
dynamics for the rigid spacecraft can be obtained as~\cite{chen1993sliding}
\begin{equation}
\left\{
\begin{tabular}{l}
$\dot{\boldsymbol{q}}_{\mathrm{e}}=\dfrac{1}{2}\left[
\begin{array}{c}
-\boldsymbol{q}_{\mathrm{ev}}^{\mathrm{T}} \\
q_{\mathrm{e}0}I_{3}+\boldsymbol{q}_{\mathrm{ev}}^{\times }%
\end{array}%
\right] \boldsymbol{\omega }_{\mathrm{e}},$ \\
$J\dot{\boldsymbol{\omega}}_{\mathrm{e}}=-\boldsymbol{\omega }_{\mathrm{e}%
}^{\times }J\boldsymbol{\omega }_{\mathrm{e}}+\boldsymbol{u}+\boldsymbol{d}.$%
\end{tabular}%
\right.  \label{system model}
\end{equation}

In addition, the error quaternion can also be written as~\cite{di2003output}%
,
\begin{equation}
\boldsymbol{q}_{\mathrm{e}}=\left[
\begin{array}{c}
q_{\mathrm{e}0} \\
\boldsymbol{q}_{\mathrm{ev}}%
\end{array}%
\right] =\left[
\begin{array}{c}
\cos \frac{\theta \left( t\right) }{2} \\
\boldsymbol{e}\sin \frac{\theta \left( t\right) }{2}%
\end{array}%
\right] ,  \label{unit attitude error}
\end{equation}%
where $\theta \left( t\right) \in \left[ 0,2\pi \right] $ is the rotation
angle and $\boldsymbol{e}\in \mathbb{R}^{3}$ is the fixed Euler axis. Then,
the following relation can be obtained by the first relation of~(\ref{unit
attitude error}),
\begin{equation}
\theta \left( t\right) =2\arccos q_{\mathrm{e}0}.  \label{theta}
\end{equation}%
It follows from the first relation of~(\ref{system model}) and~(\ref%
{unitquaternion}) that
\begin{align}
\dot{\theta}\left( t\right) & =-\frac{2\dot{q}_{\mathrm{e}0}}{\sqrt{1-{q}_{%
\mathrm{e}0}^{2}}}  \notag \\
& =\frac{\boldsymbol{q}_{\mathrm{ev}}^{\mathrm{T}}\boldsymbol{\omega }_{%
\mathrm{e}}}{\sqrt{1-{q}_{\mathrm{e}0}^{2}}}  \notag \\
& =\frac{\boldsymbol{q}_{\mathrm{ev}}^{\mathrm{T}}}{\left\Vert \boldsymbol{q}%
_{\mathrm{ev}}\right\Vert }\boldsymbol{\omega }_{\mathrm{e}}.
\label{angularEuler}
\end{align}

According to~(\ref{multi}), $q_{\mathrm{e}0}\left( 0\right) $ and $%
\boldsymbol{q}_{\mathrm{ev}}\left( 0\right) $ can be obtained as long as the
initial attitude $\boldsymbol{q}\left( 0\right) $ of $\boldsymbol{q}$ and
the desired attitude $\boldsymbol{q}_{\mathrm{d}}$ are given. Further, the
initial value $\theta\left(0\right)$ of $\theta\left(t\right)$ can be
obtained by~(\ref{theta}). By designing an attitude maneuver controller, the
rigid spacecraft is driven to rotate about the fixed Euler axis $\boldsymbol{%
e}$, such that the rotation angle $\theta\left(t\right)$ converges from the
initial value $\theta\left(0\right)$ to the equilibrium point.

\subsection{Unwinding Phenomenon}

It can be obtained from~(\ref{unit attitude error}) that $q_{\mathrm{e0}%
}|_{\theta \left( t\right) =0}=1$ and $q_{\mathrm{e0}}|_{\theta \left(
t\right) =2\pi }=-1$, while $\theta \left( t\right) =0$ and $\theta \left(
t\right) =2\pi $ represent the same position. Thus, $\boldsymbol{q}_{\mathrm{%
e}0}\!=1$ and $\boldsymbol{q}_{\mathrm{e}0}= -1$ are both the equilibrium
point of the attitude error dynamics~(\ref{system model}) for a rigid
spacecraft.
%It can be easy obtained from Fig.~\ref{fig1} that $\theta=0$ and $%
%\theta=2\pi $ represent the same position, thus they are both the
%equilibrium point of the attitude maneuver control system for a rigid
%spacecraft. But it can be obtained from equation (\ref{unit attitude error}) that $q_{\mathrm{e0}}|_{\theta=0}=1$ and $q_{\mathrm{e0}}|_{\theta=2\pi}=-1$.
However, in most existing controller design approaches, only $\boldsymbol{q}%
_{\mathrm{e}0}=1$ is considered as equilibrium point. In this case, when the
initial value of $q_{\mathrm{e0}}$ is less than $0$, the designed controller
drives $q_{\mathrm{e0}}$ to $0$, and finally to $1$. This means that the
rigid spacecraft needs to rotate a Euler angle $\theta \left( t\right) $
larger than $\pi $. This is the "unwinding phenomenon". However, the rigid
spacecraft can reach the desired attitude by rotating an angle smaller than $%
\pi $. %can reach the desired attitude.

\subsection{Control Objective}

\label{controobjective}

The control task in this work is to design an anti-unwinding attitude
controller to accomplish a rest-to-rest attitude maneuver for the rigid
spacecraft system~(\ref{dynamic equaiton1}). By adopting the designed
control law for the closed-loop attitude maneuver error dynamics~(\ref%
{system model}) of a rigid spacecraft, the following relations are achieved,
\begin{equation}
\mathrm{lim}_{t\rightarrow \infty}q_{\mathrm{e}0}=1\ \mathrm{or} -1,\
\mathrm{lim}_{t\rightarrow \infty }\boldsymbol{\omega}_{\mathrm{e}}=0.
\label{aim}
\end{equation}
Moreover, the unwinding phenomenon is avoided during the rigid spacecraft
maneuver.

\section{Controller Design}

\label{controllerdesign}

In this section, we aim to design an anti-unwinding sliding mode control law
to accomplish the control objective stated in Section~\ref{controobjective}.
%rotate the rigid spacecraft (\ref{system model}) from $%\mathcal{F}_{\mathrm{b}}$ to $\mathcal{F}_{\mathrm{d}}$ with anti-unwinding performance.
First, a new switching surface is constructed in Section~\ref%
{slidingsurfacesec}, which considers both $\boldsymbol{q}_{\mathrm{e}0}=1$
and $\boldsymbol{q}_{\mathrm{e}0}=-1$ to be equilibrium points. Then, the
anti-unwinding performance when the system states are on the switching
surface is proven. In section~\ref{equivalentsec}, an anti-unwinding sliding
mode attitude control law is derived based on the constructed switching
function. In addition, a dynamic parameter is introduced to guarantee the
anti-unwinding performance when the system states are outside the switching
surface.

Before preceding, we first give the following lemmas.

\begin{lemma}
\cite{Lemmastability} \label{lemma1}
Suppose $V(x)$ is a $C^{1}$ smooth positive-definite
function (defined on $U\subset \mathbb{R}^{n}$) and $\dot{V}(x)+\lambda V^{\alpha
}(x) $ is a negative semi-definite function on $U\subset \mathbb{R}^{n}$ for $\alpha
\in (0,1)$ and $\lambda \in \mathbb{R}^{+}$, then there exists an area $U_{0}\subset
\mathbb{R}^{n} $ such that any $V(x)$ which starts from $U_{0}\subset \mathbb{R}^{n}$ can
reach $V(x)\equiv 0$ in finite time. Moreover, if $T_{\mathrm{s}}$ is the
time needed to reach $V(x)\equiv 0$, then
\begin{equation*}
T_{\mathrm{s}}\leq \frac{V^{1-\alpha }(x_{0})}{\lambda \left( 1-\alpha
\right) },
\end{equation*}%
where $V(x_{0})$ is the initial value of $V(x)$.
\end{lemma}

\begin{lemma}
\label{idenpotent} For any unit vector $\boldsymbol{x}\in \mathbb{R}^{n},$
the matrix $A=\boldsymbol{xx}^{\mathrm{T}}$ is a $n\times n$ idempotent
matrix.
\end{lemma}

\begin{proof}
Note that $A=\boldsymbol{xx}^{\mathrm{T}}$ and $\boldsymbol{x}$ is a unit
vector, thus
\begin{align*}
AA &=\boldsymbol{xx}^{\mathrm{T}}\boldsymbol{xx}^{\mathrm{T}} \\
&=\boldsymbol{xx}^{\mathrm{T}} \\
&=A.
\end{align*}%
This implies that the matrix $A$ is an idempotent matrix.
\end{proof}

\begin{lemma}\label{eigen}
For any idempotent matrix $A\in \mathbb{R}^{n\times n}$, its eigenvalues are $1$ or $0$.
\end{lemma}

\begin{proof}
Suppose that the non-zero vector $\boldsymbol{y}\in \mathbb{R}^{n}$ is an
eigenvector corresponding to a non-zero eigenvalue $\lambda $ of the matrix $%
A$. Then, we have
\begin{equation}
A\boldsymbol{y}=\lambda \boldsymbol{y}\text{.}  \label{A1}
\end{equation}%
Multiplying $A$ of both sides of the above relation, gives%
\begin{equation}
AA\boldsymbol{y}=\lambda A\boldsymbol{y}\text{.}  \label{A2}
\end{equation}%
Because $A$ is an idempotent matrix, thus the left side of (\ref{A2}) can be
rewritten as $A\boldsymbol{y}$. In addition, by (\ref{A1}), the right side
of (\ref{A2}) can be rewritten as $\lambda ^{2}\boldsymbol{y}$. Then, there
holds
\begin{equation}
A\boldsymbol{y}=\lambda ^{2}\boldsymbol{y}.  \label{A3}
\end{equation}
Combining (\ref{A1}) with (\ref{A3}), yields
\begin{equation*}
\lambda \boldsymbol{y}=\lambda ^{2}\boldsymbol{y}\text{.}
\end{equation*}%
Because the vector $\boldsymbol{y}$ is a non-zero vector, then there holds $%
\lambda =1$ or $0$. Thus, the proof is completed.
\end{proof}

\subsection{Switching Surface}

\label{slidingsurfacesec} For the attitude error dynamics~(\ref{system model}%
) of a rigid spacecraft,
%we aim to construct a sliding mode attitude control law with anti-unwinding performance. For this end,
we design the following switching function,%
\begin{equation}
\boldsymbol{s}=\boldsymbol{\omega }_{\mathrm{e}}+\lambda \boldsymbol{\sigma }%
,  \label{switching surface}
\end{equation}%
where $\lambda $ is a positive constant, and
\begin{equation}
\boldsymbol{\sigma }:=\sinh \left( q_{\mathrm{e}0}\right) \boldsymbol{q}_{%
\mathrm{ev}}.  \label{e_r1}
\end{equation}
\quad Next, the fact that the switching surface $\boldsymbol{s}=0$
containing two equilibriums $\boldsymbol{q}_{\mathrm{e}0}=1$ and $-1$ is
proven. In addition, the convergence performance of the attitude error
variables $\boldsymbol{\omega }_{\mathrm{e}}$ and $\boldsymbol{q}_{\mathrm{ev%
}}$ on the switching surface $\boldsymbol{s}=0$ is analyzed. In addition,
the anti-unwinding performance of the designed switching function~(\ref%
{switching surface}) in the sliding phase is demonstrated.

Before given the theorem, we should give some properties of the functions $%
\cosh q_{\mathrm{e}0}$ and $\sinh \cos \frac{\theta \left( t\right) }{2}$.
The maximum value of the function $\cosh q_{\mathrm{e}0}$ can be obtained
when $q_{\mathrm{e}0}=1$ and $q_{\mathrm{e}0}=-1$. For $\theta \left(
t\right) \in \left( 0,\pi \right] $, $\sinh \cos \frac{\theta \left(
t\right) }{2}\geq 0$, and for $\theta \left( t\right) \in \left( \pi ,2\pi
\right) $, $\sinh \cos \frac{\theta \left( t\right) }{2}\leq 0$.

\begin{theorem}\label{aucoss}
If the system states of the attitude error dynamics~(\ref{system model}) are restricted to the switching surface $\boldsymbol{s}%
=0$, the following conclusions are achieved:

(i) The switching surface $\boldsymbol{s}%
=0$ contains two equilibriums $\boldsymbol{q}_{\mathrm{e}0}=1$ and $-1$, and the control goal in~(\ref{aim}) is guaranteed.

(ii) The unwinding
phenomenon is avoided in the sliding phase.%
%, that is $q_{\mathrm{e}0}$ converges to $1$ or $-1$.
% by adopting the
%designed attitude error and attitude error vector.switching surface $s$ converge to $0$.
\end{theorem}

\begin{proof}
First, we choose the following Lyapunov function,%
\begin{equation}
V_{1}\left( t\right) :=2\left( \kappa-\cosh q_{\mathrm{e}0} \right) ,
\label{v1}
\end{equation}%
where $\kappa=\max\left(\cosh q_{\mathrm{e}0}\right)$ for $q_{\mathrm{e}%
0}\in \left[-1,\ 1\right].$ By taking time derivative of~(\ref{v1}), and
using the first equation of~(\ref{system model}), the condition $\boldsymbol{%
s}=0$, and~(\ref{e_r1}), we have
\begin{align}
\dot{V}_{1}\left( t\right) & =-2\sinh \left( q_{\mathrm{e}0}\right)\dot{q}_{%
\mathrm{e}0}  \notag \\
& =\sinh \left( q_{\mathrm{e}0}\right) \boldsymbol{q}_{\mathrm{ev}}^{\mathrm{%
T}}\boldsymbol{\omega }_{\mathrm{e}}  \notag \\
& =\boldsymbol{\sigma }^{\mathrm{T}}\boldsymbol{\omega }_{\mathrm{e}}  \notag
\\
& =-\lambda \boldsymbol{\sigma }^{\mathrm{T}}\boldsymbol{\sigma } .
\label{dotv3}
\end{align}
Thus, $\dot{V}_{1}\left(t\right)\leq0$. Further, it can be derived from~(\ref%
{dotv3}) that if $\dot{V}_{1}\left( t\right)= 0$, there holds $\boldsymbol{%
\sigma}=0$. Then, it follows from~(\ref{e_r1}) that $q_{\mathrm{e}0}=0$ or $%
\boldsymbol{q}_{\mathrm{ev}}=0$. According to~(\ref{unitquaternion}), there
holds $q_{\mathrm{e}0}=1$ or $-1$ when $\boldsymbol{q}_{\mathrm{ev}}=0$.
Moreover, it can be obtained from~(\ref{v1}) that $\min
\left(V_{1}\left(t\right)\right)=V_{1}\left(t\right)|_{q_{\mathrm{e}%
0}=1}=V_{1}\left(t\right)|_{q_{\mathrm{e}0}=-1}=0$, and $V_{1}\left(t%
\right)|_{q_{\mathrm{e}0}=0}\neq 0$. This means that the switching surface $%
\boldsymbol{s}=0$ contains two equilibriums $\boldsymbol{q}_{\mathrm{e}0}=1$
and $\boldsymbol{q}_{\mathrm{e}0}=-1$. In addition, substituting $%
\boldsymbol{q}_{\mathrm{ev}}=0$ into~(\ref{switching surface}) gives $%
\boldsymbol{\omega}_{\mathrm{e}}=0$.

Thus, the conclusion (i) is proven.
%and $\boldsymbol{q}_{\mathrm{e}}=\left[1\ 0\ 0\ 0\right]^{\mathrm{T}}$ or $\boldsymbol{q}_{\mathrm{e}}=\left[-1\ 0\ 0\ 0\right]^{\mathrm{T}}$. and when $\dot{V}_{1}\left( t\right)= 0$, there holds $\boldsymbol{\sigma }=0$. Further, it can be derived from (\ref{e_r1}) that
%As $\dot{V}_{1}\left( t\right) \leq 0$, which means $V_{1}\left( t\right) $
%is a non-increasing positive valued function. Thus, $V_{1}\left( t\right) $ is
%bounded, and it can be obtained from the properties of function $\cosh q_{\mathrm{e}0}$ that $\left(q_{\mathrm{e}0}-1\right)$ or $\left(q_{\mathrm{e}%
%0}+1\right)$ are bounded. Then, it can be obtained from~(\ref%
%{unitquaternion}) and~(\ref{v1}) that $\boldsymbol{q}_{\mathrm{ev}}$ is
%bounded. Further, according to~(\ref{switching surface}) and~(\ref{e_r1}), $\boldsymbol{\omega }_{\mathrm{e}}$ is also bounded.
% According to (\ref{v1}), it can be concluded that $\boldsymbol{\omega }_{\mathrm{e}}$ and $\boldsymbol{q}_{\mathrm{ev}}$ are bounded.
%Therefore, the attitude error based rigid spacecraft system (\ref{systemmodel}) is uniformly stable.

Next, the anti-unwinding performance of the attitude error dynamics~(\ref%
{system model}) with system states being on the switching surface $%
\boldsymbol{s}=0 $ is proven. According to~(\ref{theta}), the Lyapunov
function (\ref{v1}) can be rewritten as
\begin{equation*}
V_{1}\left( t\right) :=2\left( \kappa -\cosh \cos \frac{\theta \left(
t\right) }{2}\right) .
\end{equation*}%
Consequently,
\begin{equation}
\dot{V}_{1}\left( t\right) =\sin \frac{\theta \left( t\right) }{2}\sinh \cos
\frac{\theta \left( t\right) }{2}\dot{\theta}\left( t\right) .  \label{dotv1}
\end{equation}%
In addition, there hold $\sin \frac{\theta \left( t\right) }{2}>0$ for $%
\theta \left( t\right) \in \left( 0,2\pi \right) $, $\sinh \cos \frac{\theta
\left( t\right) }{2}\geq 0$ for $\theta \left( t\right) \in \left( 0,\pi %
\right] $, and $\sinh \cos \frac{\theta \left( t\right) }{2}\leq 0$ for $%
\theta \left( t\right) \in \left( \pi ,2\pi \right) $. Note that $\dot{V}%
_{1}\left( t\right) \leq 0$, it can be derived from~(\ref{dotv1}) that there
hold $\dot{\theta}\left( t\right) \leq 0$ for $\theta \left( t\right) \in
\left( 0,\pi \right] $ and $\dot{\theta}\left( t\right) \geq 0$ for $\theta
\left( t\right) \in \left( \pi ,2\pi \right) $. Suppose that the system
states reach the switching surface $\boldsymbol{s}=0$ when $t=t_{s0}$. Then,
if $\theta \left( t_{s0}\right) \in \left( 0,\pi \right] $, there holds $%
\lim_{t\rightarrow \infty }\theta \left( t\right) =0$, and if $\theta \left(
t_{s0}\right) \in \left( \pi ,2\pi \right) $, there holds $%
\lim_{t\rightarrow \infty }\theta \left( t\right) =2\pi $. This implies that
the unwinding phenomenon is avoided when the system states are restricted to
the switching surface $\boldsymbol{s}=0$.
\end{proof}

\subsection{Anti-Unwinding Sliding Mode Attitude Maneuver Control Law}

\label{equivalentsec}

In this section, we need to construct a control law such that the condition $%
\boldsymbol{s}^{\mathrm{T}}\boldsymbol{\dot{s}}<0$ is satisfied. This
condition assures us that the switching surface $\boldsymbol{s}=0$ will
attract all the system trajectories.

Consider a class of state feedback control for the attitude error dynamics~(%
\ref{system model}) of a rigid spacecraft in the following form,
\begin{equation}
\boldsymbol{u}=\boldsymbol{u}_{\mathrm{eq}}+\boldsymbol{u}_{\mathrm{n}},
\label{control law}
\end{equation}%
where the term $\boldsymbol{u}_{\mathrm{eq}}$ is the equivalent control for
the nominal system, the term $\boldsymbol{u}_{\mathrm{n}}$ is designed to
compensate the disturbance.
%, and the term $%\boldsymbol{u}_{\mathrm{f}}$ is added to guarantee the finite-time convergence of the switching function $\boldsymbol{s}$.
Thus, the equivalent control $\boldsymbol{u}_{\mathrm{eq}}$ can be obtained
from the nominal system part by setting $\dot{\boldsymbol{s}}$ to be zero.
That is%
\begin{equation}
\dot{\boldsymbol{s}}=\dot{\boldsymbol{\omega}}_{\mathrm{e}}+\lambda \dot{%
\boldsymbol{\sigma}}=0.  \label{dots}
\end{equation}%
The nominal part of the attitude error dynamics~(\ref{system model}) is
\begin{equation*}
\dot{\boldsymbol{\omega}}_{\mathrm{e}}=J^{-1}\left( -\boldsymbol{\omega }_{%
\mathrm{e}}^{\times }J\boldsymbol{\omega }_{\mathrm{e}}+\boldsymbol{u}_{%
\mathrm{eq}}\right) .
\end{equation*}%
Substituting this expression into~(\ref{dots}), gives%
\begin{equation}
\boldsymbol{u}_{\mathrm{eq}}=\boldsymbol{\omega }_{\mathrm{e}}^{\times }J%
\boldsymbol{\omega }_{\mathrm{e}}-\lambda J\dot{\boldsymbol{\sigma}}.
\label{u_eq}
\end{equation}%
The control term $\boldsymbol{u}_{\mathrm{n}}$ is designed as%
\begin{equation}
\boldsymbol{u}_{\mathrm{n}}=-\left(\gamma
_{1}+\gamma_{2}\left(t\right)\right)\boldsymbol{f}\left( s\right) ,
\label{u_n}
\end{equation}%
where $\gamma _{1}\geq \left\Vert \boldsymbol{d}\right\Vert _{\max }$, $%
\gamma_{2}\left(t\right)$ is a positive-valued function which will be given
later, and
\begin{equation}
\boldsymbol{f}\left( \boldsymbol{s}\right) =\left\{
\begin{array}{l}
\mathrm{sgn}\left( \boldsymbol{s}\right) ,\left\Vert \boldsymbol{s}%
\right\Vert \neq 0, \\
0,\qquad \left\Vert \boldsymbol{s}\right\Vert =0,%
\end{array}%
\right.  \label{f}
\end{equation}%
with $\mathrm{sgn}\left( \boldsymbol{s}\right) =\left[ \mathrm{sgn}\left(
s_{1}\right) \;\mathrm{sgn}\left( s_{2}\right) \;\mathrm{sgn}\left(
s_{3}\right) \right] ^{\mathrm{T}}$, and
\begin{equation*}
\mathrm{sgn}\left( s_{i}\right) =\left\{
\begin{array}{c}
1,\text{ }s_{i}>0 \\
-1,\text{ }s_{i}\leq 0%
\end{array}%
\right. \left( i=1,2,3\right) .
\end{equation*}

Then, the following anti-unwinding sliding mode attitude maneuver control
(briefly, AUSMAMC) law is presented,%
\begin{equation}
\left\{
\begin{array}{l}
\boldsymbol{u}=\boldsymbol{u}_{\mathrm{eq}}+\boldsymbol{u}_{\mathrm{n}}, \\
\boldsymbol{u}_{\mathrm{eq}}=\boldsymbol{\omega }_{\mathrm{e}}^{\times }J%
\boldsymbol{\omega }_{\mathrm{e}}-\lambda J\dot{\boldsymbol{\sigma}}, \\
\boldsymbol{u}_{\mathrm{n}}=-\left(\gamma
_{1}+\gamma_{2}\left(t\right)\right)\boldsymbol{f}\left( s\right) , \\
\boldsymbol{s}=\boldsymbol{\omega }_{\mathrm{e}}+\lambda \boldsymbol{\sigma }%
, \\
\boldsymbol{\sigma }=\sinh \left( q_{\mathrm{e}0}\right) \boldsymbol{q}_{%
\mathrm{ev}},%
\end{array}%
\right.  \label{sliding mode control}
\end{equation}%
where $\lambda $ is a positive numbers, $\gamma _{1}\geq\left\Vert
\boldsymbol{d}\right\Vert _{\max }$, and $\gamma_{2}\left(t\right)$ is a
positive-valued function, which will be given in the following section.
%Moreover, the convergence performance of the system states of the rigid
%spacecraft (\ref{system model}) under the ANSMCAMC controller (\ref{sliding
%mode control}) is given.

\subsection{Convergence Analysis}

\label{convergenceanalysis} %Before preceding, some preliminaries are first
%given, which will be utilized in the convergence analysis of the proposed
%SMAMCAU law (\ref{sliding mode control}).

%\begin{lemma}
%\cite{Lemma1} \label{Lemma1} Suppose $V\left( x\right) $ is a $C^{1}$ smooth
%positive-definite function (defined on $U\subset R^{n}$) and $\dot{V}\left(
%x\right) +\gamma V^{\alpha }\left( x\right) $ is a negative semi-definite
%function on $U\subset R^{n}$ for $\alpha \in \left( 0,1\right) $\ and $%
%\gamma \in R^{+}$, then there exists an area $U_{0}\subset R^{n}$ such that
%any $V\left( x\right) $ which starts from $U\subset R^{n}$ can reach \ in
%finite time. Moreover, if $t_{\mathrm{s}}$ is the time needed to reach $%
%V\left( x\right) \equiv 0$, then there holds
%\begin{equation*}
%t_{\mathrm{s}}\leq \frac{V^{1-\alpha }\left( x_{0}\right) }{\gamma \left(
%1-\alpha \right) },
%\end{equation*}%
%where $V\left( x_{0}\right) $ is the initial value of $V\left( x\right) .$
%\end{lemma}

%Next, the convergence of the closed-loop system under the AUSMAMC law (\ref%
%{sliding mode control}) is analyzed. In addition, the anti-unwinding
%performance is also proven in the following theorem.
In this section, the convergence of the closed-loop system under the
developed AUSMAMC law (\ref{sliding mode control}) is analyzed. In addition,
the anti-unwinding performance is proven in the following theorem.
\begin{theorem}\label{aubss}
Consider a rigid spacecraft described by~(\ref{system model}) in the
presence of disturbance. If the parameter $%
\gamma_{2}\left(t\right)$ of the proposed AUSMAMC law (\ref{sliding mode control}%
) is chosen as
\begin{equation}
\gamma _{2}\left( t\right) =\frac{%
\lambda }{\lambda _{\min }\left(J^{-1}\right)}\left\vert \dot{g} \right\vert ,  \label{gamma2}
\end{equation}%
where \begin{equation}
g:=\sinh \left( q_{\mathrm{e}0}\right) \left\Vert \boldsymbol{q}_{\mathrm{ev}%
}\right\Vert.\label{g}
\end{equation}%
 Then, the following
conclusions are achieved:

(i) The switching function $\boldsymbol{s}$ converges to zero in finite time.

(ii) The unwinding phenomenon is avoided before the system states reach the switching surface $\boldsymbol{s}=0$.%during the reaching phase.
%the new designed attitude error and attitude error vector.

%(ii) When the system state is restricted to the switching surface $s=0$,
%then the attitude variables $\omega $ and $q_{\mathrm{v}}$ converge to $0$.
%Furthermore, the anti-unwinding performance can be guaranteed during this
%sliding phase.% by adopting the
%designed attitude error and attitude error vector.
\end{theorem}

\begin{proof}
To prove the conclusion (i), we choose the following Lyapunov function,
\begin{equation}
V_{2}\left( t\right) =\frac{1}{2}\boldsymbol{s}^{\mathrm{T}}\boldsymbol{s}.
\label{V3}
\end{equation}%
With the help of~(\ref{system model}) and~(%
\ref{switching surface}), we obtain from~(\ref{V3}) that
\begin{align}
\dot{V}_{2}\left( t\right) & =\boldsymbol{s}^{\mathrm{T}}\boldsymbol{\dot{s}}
\notag \\
& =\boldsymbol{s}^{\mathrm{T}}\left( \boldsymbol{\dot{\omega}}_{\mathrm{e}%
}+\lambda \boldsymbol{\dot{\sigma}}\right)  \notag \\
& =\boldsymbol{s}^{\mathrm{T}}\left( J^{-1}\left( -\boldsymbol{\omega }_{%
\mathrm{e}}^{\times }J\boldsymbol{\omega }_{\mathrm{e}}+\boldsymbol{u}+%
\boldsymbol{d}\right) +\lambda \dot{\boldsymbol{\sigma}}\right) .
\label{v1zhongjian}
\end{align}%
Substituting the AUSMAMC law (\ref{sliding mode control}) with $\gamma
_{1}\geq \left\Vert \boldsymbol{d}\right\Vert _{\max }$ into~(\ref%
{v1zhongjian}), results in
\begin{align}
\dot{V}_{2}\left( t\right) & =\boldsymbol{s}^{\mathrm{T}}\left( J^{-1}\left(
-\boldsymbol{\omega }_{\mathrm{e}}^{\times }J\boldsymbol{\omega }_{\mathrm{e}%
}+\boldsymbol{u}_{\mathrm{eq}}+\boldsymbol{u}_{\mathrm{n}}+\boldsymbol{d}%
\right) +\lambda \boldsymbol{\dot{\sigma}}\right)  \notag \\
& =\boldsymbol{s}^{\mathrm{T}}J^{-1}\left( \boldsymbol{u}_{\mathrm{n}}+%
\boldsymbol{d}\right)  \notag \\
& =-\gamma _{2}\left( t\right) \boldsymbol{s}^{\mathrm{T}}J^{-1}\text{sgn}%
\left( \boldsymbol{s}\right) +\boldsymbol{s}^{\mathrm{T}}J^{-1}\left(
\left\Vert \boldsymbol{d}\right\Vert _{\max }-\gamma _{1}\right)  \notag \\
& \leq -\gamma _{2}\left( t\right) \boldsymbol{s}^{\mathrm{T}}J^{-1}\text{sgn%
}\left( \boldsymbol{s}\right) .  \label{sleq0}
\end{align}%
It is obvious that there holds%
\begin{equation}
\boldsymbol{s}^{\mathrm{T}}J^{-1}\mathrm{sgn}\left( \boldsymbol{s}\right)
\geq \lambda _{\min }\left( J^{-1}\right) \left\Vert \boldsymbol{s}%
\right\Vert .  \label{inequ}
\end{equation}%
Thus, it can be obtained from (\ref{sleq0}) and (\ref{inequ}) that
\begin{equation*}
\dot{V}_{2}\left( t\right) \leq -\gamma _{2}\left( t\right) \lambda _{\min
}\left( J^{-1}\right) \left\Vert \boldsymbol{s}\right\Vert .
\end{equation*}%
By combining this with (\ref{V3}) and~(\ref{inequ}), it is easy to obtain
that%
\begin{align}
\dot{V}_{2}\left( t\right) & \leq -\sqrt{2}\gamma _{2}\left( t\right)
\lambda _{\min }\left( J^{-1}\right) \left( \frac{1}{2}\boldsymbol{s}^{%
\mathrm{T}}\boldsymbol{s}\right) ^{\frac{1}{2}}  \notag \\
& =-\sqrt{2}\gamma _{2}\left( t\right) \lambda _{\min }\left( J^{-1}\right)
V_{2}^{\frac{1}{2}}\left( t\right) .  \label{V3finite}
\end{align}%
Clearly, $ \dot{V}_{2}\left( t\right)\leq 0$. Thus, it can be obtained from~Lemma \ref{lemma1} that the switching function
$\boldsymbol{s}$ converges to $0$ in finite time. The proof of (i) is
completed.

Next, by designing the dynamic parameter $\gamma_{2}\left(t\right)$ for the
proposed AUSMAMC law~(\ref{sliding mode control}), the anti-unwinding
performance before the system states reach the switching surface $%
\boldsymbol{s}=0$ is guaranteed.

In view of~(\ref{V3finite}), we get%
\begin{equation*}
\frac{\dot{V}_{2}\left( t\right) }{V_{2}^{\frac{1}{2}}\left( t\right) }\leq -%
\sqrt{2}\gamma _{2}\left( t\right) \lambda _{\min }\left( J^{-1}\right) .
\end{equation*}%
Suppose that the initial time is $t_{0}=0$. Then, by taking integral of both
sides of the above equation, we arrive at%
\begin{equation*}
\int_{0}^{t}\frac{\dot{V}_{2}\left( \tau \right) }{V_{2}^{\frac{1}{2}}\left(
\tau \right) }d\tau \leq -\sqrt{2}\lambda _{\min }\left( J^{-1}\right)
\int_{0}^{t}\gamma _{2}\left( \tau \right) d\tau .
\end{equation*}%
A direct calculation gives
\begin{equation}
V_{2}^{\frac{1}{2}}\left( t\right) \leq -\frac{\lambda _{\min }\left(
J^{-1}\right) }{\sqrt{2}}\int_{0}^{t}\gamma _{2}\left( \tau \right) d\tau
+V_{2}^{\frac{1}{2}}\left( 0\right) .  \label{V3inequality}
\end{equation}

Let%
\begin{equation}
v\left( t\right) =\frac{\boldsymbol{q}_{\mathrm{ev}}^{\mathrm{T}}}{%
\left\Vert \boldsymbol{q}_{\mathrm{ev}}\right\Vert }\boldsymbol{s}.
\label{vts}
\end{equation}%
By~(\ref{angularEuler}),~(\ref{switching surface}), and~(\ref{g}), the above
relation can be rewritten as%
\begin{align}
v\left( t\right) & =\left( \frac{\boldsymbol{q}_{\mathrm{ev}}^{\mathrm{T}}}{%
\left\Vert \boldsymbol{q}_{\mathrm{ev}}\right\Vert }\boldsymbol{\omega }_{%
\mathrm{e}}+\lambda \sinh \left( q_{\mathrm{e}0}\right) \frac{\boldsymbol{q}%
_{\mathrm{ev}}^{\mathrm{T}}\boldsymbol{q}_{\mathrm{ev}}}{\left\Vert
\boldsymbol{q}_{\mathrm{ev}}\right\Vert }\right)   \notag \\
& =\left( \dot{\theta}\left( t\right) +\lambda g\right) .  \label{vt}
\end{align}%
Further, it can be derived from (\ref{vts}) that%
\begin{align}
v^{2}\left( t\right) & =\left( \frac{\boldsymbol{q}_{\mathrm{ev}}^{\mathrm{T}%
}}{\left\Vert \boldsymbol{q}_{\mathrm{ev}}\right\Vert }\boldsymbol{s}\right)
^{\mathrm{T}}\boldsymbol{\frac{\boldsymbol{q}_{\mathrm{ev}}^{\mathrm{T}}}{%
\left\Vert \boldsymbol{q}_{\mathrm{ev}}\right\Vert }s}\notag \\
& =\boldsymbol{s}^{\mathrm{T}}\frac{\boldsymbol{q}_{\mathrm{ev}}\boldsymbol{q%
}_{\mathrm{ev}}^{\mathrm{T}}}{\left\Vert \boldsymbol{q}_{\mathrm{ev}%
}\right\Vert ^{2}}\boldsymbol{s}.\label{v2t}
\end{align}%
Note that $\frac{\boldsymbol{q}_{\mathrm{ev}}}{\left\Vert \boldsymbol{q}_{\mathrm{ev%
}}\right\Vert }$ is a unit vector, thus it follows from~(\ref{v2t}), Lemmas~\ref{idenpotent} and~{\ref{eigen}} that%
\begin{align*}
v^{2}\left( t\right) & \leq \lambda _{\max }\left( \frac{\boldsymbol{q}_{%
\mathrm{ev}}\boldsymbol{q}_{\mathrm{ev}}^{\mathrm{T}}}{\left\Vert
\boldsymbol{q}_{\mathrm{ev}}\right\Vert ^{2}}\right) \left\Vert \boldsymbol{s%
}\right\Vert ^{2} \notag\\
& \leq \left\Vert \boldsymbol{s}\right\Vert ^{2}.
\end{align*}%
Combining this with (\ref{V3}) yields%
\begin{equation}
\frac{1}{2}v^{2}\left( t\right) \leq V_{2}\left( t\right) .  \label{vtV2}
\end{equation}

It should be noted that, the rest-to-rest attitude maneuver issue is
considered in this paper. Thus, $\boldsymbol{\omega }_{\mathrm{e}}\left(
0\right) =0$. Further, it can be obtained from~(\ref{angularEuler}) that $\dot{\theta}%
\left( 0\right) =0$. Then, the initial value of $v\left( t\right) $ in (\ref%
{vt}) can be obtained as
\begin{equation}
v\left( 0\right) =\lambda g\left( 0\right) .  \label{v0}
\end{equation}%
Moreover, with the help of~(\ref{g}) and~(\ref{v0}), we can get the initial value of $%
V_{2}\left( 0\right) $ in (\ref{V3}) as (recall that $\boldsymbol{\omega}_{\mathrm{e}}=0$)
\begin{align}
V_{2}\left( 0\right) & =\frac{1}{2}\lambda ^{2}\left( \sinh \left( q_{%
\mathrm{e}0}\right) \boldsymbol{q}_{\mathrm{ev}}\right) ^{\mathrm{T}}\left(
\sinh \left( q_{\mathrm{e}0}\right) \boldsymbol{q}_{\mathrm{ev}}\right)
\notag \\
& =\frac{1}{2}\lambda ^{2}\sinh ^{2}\left( q_{\mathrm{e}0}\right)
\boldsymbol{q}_{\mathrm{ev}}^{\mathrm{T}}\boldsymbol{q}_{\mathrm{ev}}  \notag
\\
& =\frac{1}{2}v^{2}\left( 0\right)  \label{V20}
\end{align}%
Thus, the following relation can be obtained from~(\ref{V3inequality}),~(\ref%
{vtV2}), and~(\ref{V20}),
\begin{align*}
\left( \frac{1}{2}v^{2}\left( t\right) \right) ^{\frac{1}{2}}\leq V_{2}^{%
\frac{1}{2}}\left( t\right) \leq & -\frac{\lambda _{\min }\left(
J^{-1}\right) }{\sqrt{2}}\int_{0}^{t}\gamma _{2}\left( \tau \right) d\tau \\
& +\left( \frac{1}{2}v^{2}\left( 0\right) \right) ^{\frac{1}{2}},
\end{align*}%
which can be further written as,
\begin{equation}
|v\left( t\right) |\leq -\lambda _{\min }\left( J^{-1}\right)
\int_{0}^{t}\gamma _{2}\left( \tau \right) d\tau +|v\left( 0\right) |.
\label{inter1}
\end{equation}%
Moreover, because $\gamma _{2}\left( t\right) > 0$, then it can be
obtained from~(\ref{inter1}) that $v\left( t\right) $ will decrease to $0$
when $v\left( 0\right) >0$, and $v\left( t\right) $ will increase to $0$
when $v\left( 0\right) <0$.
%Thus, the proof process is broken into the following two phases.

To prove the anti-unwinding property of the proposed control law (\ref%
{sliding mode control}), we need to prove that $\dot{\theta}%
\left(t\right)\leq0$ for $\theta\left(0\right) \in \left(0,\ \pi\right]$ and
$\dot{\theta}\left(t\right)\geq0$ for $\theta\left(0\right) \in \left(\pi,\
2\pi\right)$. For this end, the following two cases are considered to
complete the proof.

(a) When $\theta \left( 0\right) \in \left( 0,\pi \right] ,$ by the first
equation of~(\ref{unit attitude error}), we have $\boldsymbol{q}_{\mathrm{e}%
0}\left( 0\right) >0$. Then, using~(\ref{v0}) and~(\ref{g}), we get
\begin{equation*}
v\left( 0\right) =\lambda \sinh \left( q_{\mathrm{e}0}\left( 0\right)
\right) \left\Vert \boldsymbol{q}_{\mathrm{ev}}\left( 0\right) \right\Vert
>0.
\end{equation*}%
Thus, $v\left( t\right) $ will decrease to $0$ due to (\ref{inter1}). In
such a case, it can be further obtained from (\ref{inter1}) that
\begin{equation*}
\dot{\theta}\left( t\right) +\lambda g\leq -\lambda _{\min }\left(
J^{-1}\right) \int_{0}^{t}\gamma _{2}\left( \tau \right) d\tau +\lambda
g\left( 0\right) .
\end{equation*}%
It can be further rewritten as
\begin{align}
\dot{\theta}\left( t\right) \leq & -\lambda _{\min }\left( J^{-1}\right)
\int_{0}^{t}\gamma _{2}\left( \tau \right) d\tau +\lambda \left( g\left(
0\right) -g\right)  \notag \\
=& -\lambda _{\min }\left( J^{-1}\right) \int_{0}^{t}\gamma _{2}\left( \tau
\right) d\tau -\lambda \int_{0}^{t}\frac{\mathrm{d}g}{\mathrm{d}\tau }d\tau
\notag \\
=& -\int_{0}^{t}\left( \lambda _{\min }\left( J^{-1}\right) \gamma
_{2}\left( \tau \right) +\lambda \frac{\mathrm{d}g}{\mathrm{d}\tau }\right)
d\tau .  \label{theta11}
\end{align}

If $\dot{g}>0,$ then it can be obtained from (\ref{gamma2}) that $\gamma
_{2}\left( t\right) =\frac{\lambda \dot{g}}{\lambda _{\min }\left(
J^{-1}\right) }$. It is followed from~(\ref{theta11}) that
\begin{align*}
\dot{\theta}\left( t\right) & \leq -2\lambda \int_{0}^{t}\frac{\mathrm{d}%
g\left( \tau \right) }{\mathrm{d}\tau }d\tau \\
& \leq 0.
\end{align*}%
\quad If $\dot{g}\leq 0,$ then it can be obtained from (\ref{gamma2}) that $%
\gamma _{2}\left( t\right) =-\frac{\lambda \dot{g}}{\lambda _{\min }\left(
J^{-1}\right) }$. With this, it can be derived from (\ref{theta11}) that $%
\dot{\theta}\left( t\right) \leq 0.$

In conclusion, it can be obtained from above two cases that when $\theta
\left( 0\right) \in \left( 0,\pi \right] ,$ the rotation angle $\theta
\left(t\right)$ will decrease to $0$.

(b) When $\theta \left( 0\right) \in \left( \pi ,\ 2\pi \right) ,$ by the
first equation of~(\ref{unit attitude error}), we have $\boldsymbol{q}_{%
\mathrm{e}0}\left( 0\right) <0$. Then, using~(\ref{v0}) and~(\ref{g}), we
get
\begin{equation*}
v\left( 0\right) =\lambda \sinh \left( q_{\mathrm{e}0}\left( 0\right)
\right) \left\Vert \boldsymbol{q}_{\mathrm{ev}}\left( 0\right) \right\Vert
<0.
\end{equation*}%
Thus, $v\left( t\right) $ will increase to $0$ due to~(\ref{inter1}). In
this case, it can be obtained from~(\ref{inter1}) that,
\begin{equation*}
-\dot{\theta}\left( t\right) -\lambda g\leq -\lambda _{\min }\left(
J^{-1}\right) \int_{0}^{t}\gamma _{2}\left( \tau \right) d\tau -\lambda
g\left( 0\right) ,
\end{equation*}%
or, equivalently, %
\begin{align}
\dot{\theta}\left( t\right) & \geq \int_{0}^{t}\lambda _{\min }\left(
J^{-1}\right) \gamma _{2}\left( \tau \right) d\tau +\lambda g\left( 0\right)
-\lambda g  \notag \\
& =\int_{0}^{t}\lambda _{\min }\left( J^{-1}\right) \gamma _{2}\left( \tau
\right) d\tau -\lambda \int_{0}^{t}\frac{\mathrm{d}g}{\mathrm{d}\tau }d\tau
\notag \\
& =\int_{0}^{t}\left( \lambda _{\min }\left( J^{-1}\right) \gamma _{2}\left(
\tau \right) d\tau -\lambda \frac{\mathrm{d}g}{\mathrm{d}\tau }\right) d\tau
.  \label{theta12}
\end{align}

If $\dot{g}>0$, there holds $\gamma _{2}\left( t\right) =\frac{\lambda \dot{g%
}}{\lambda _{\min }\left( J^{-1}\right) }.$ Substitute it into (\ref{theta12}%
), we have $\dot{\theta}\left( t\right) \geq 0$.

If $\dot{g}\leq 0$, there holds $\gamma _{2}\left( t\right) =-\frac{\lambda
\dot{g}}{\lambda _{\min }\left( J^{-1}\right) }.$ Substitute it into (\ref%
{theta12}), yields%
\begin{align*}
\dot{\theta}\left( t\right) & \geq -2\lambda \int_{0}^{t}\frac{\mathrm{d}g}{%
\mathrm{d}\tau }d\tau \\
& \geq 0.
\end{align*}

Thus, it can be obtained from above two cases that when $\theta \left(
0\right) \in \left( \pi ,2\pi \right) ,$ the rotation angle $\theta
\left(t\right)$ will increase to $2\pi $.

Based on above discussion, we have proven the conclusion that the unwinding
phenomenon is successfully avoided under the AUSMAMC law~(\ref{sliding
mode control}) with $\gamma _{2}\left( t\right) =\frac{\lambda \left\vert
\dot{g}\right\vert }{\lambda _{\min }\left( J^{-1}\right) }$.
\end{proof}

In Theorem \ref{aubss}, the anti-unwinding performance before the system
states reach the switching surface is proven. In Theorem \ref{aucoss}, the
anti-unwinding performance when the system states are constricted on the
switching surface is also shown. The results in these two theorems have
illustrated that the proposed AUSMAMC law~(\ref{sliding mode control}) has
the performance of anti-unwinding.
%The anti-unwinding performance before system states reach on the switching surface $\boldsymbol{s}=0$ and are constricted on the switching surface are shown in Theorems \ref{aubss} and \ref{aucoss}, respectively.

\begin{remark}
A drawback of the control law~(\ref{u_n}) is that it is
discontinuous about the switching surface $\boldsymbol{s}=0$. This
characteristic may cause an undesirable chattering phenomenon. For practical
implementations, the controller must be smoothed. Thus, the discontinuous
function $\mathrm{sgn}\left( \boldsymbol{s}\right) $ is replaced by the smooth
continuous function $\boldsymbol{l}\left( \boldsymbol{s}\right) :=\left[
l\left( s_{1}\right) \ l\left( s_{2}\right) \ l\left( s_{3}\right) \right] ^{%
\mathrm{T}}$ with $l\left( s_{i}\right) $ in the following
equation,
\begin{equation}
l\left( s_{i}\right) :=\left\{
\begin{array}{lc}
\mathrm{sgn}\left( s_{i}\right) , & \mbox{if }\left\vert s_{i}\right\vert
\geq \varepsilon, \\
\mathrm{\arctan } \frac{s_{i}\tan \left( 1\right)} {\varepsilon} , & \mbox{if }%
\left\vert s_{i}\right\vert < \varepsilon,%
\end{array}%
i=1,2,3,\right.   \label{f_si}
\end{equation}%
where $\varepsilon $ is a small positive value. As $\varepsilon $ approaches
zero, the performance of this boundary layer can be made
arbitrarily close to that of original control law.
\end{remark}

The advantage of the proposed AUSMAMC law (\ref{sliding mode control}) is
that the unwinding phenomenon can be avoided during the rigid spacecraft
attitude maneuver, and the disturbance can be compensated by the designed
controller. Besides, the developed control law has only two tunable
parameters.
%In addition, the presented AUSMAMC law (\ref{sliding mode control}) is smooth about $\boldsymbol{s}=0$, and thus the chattering phenomenon can be avoided.

\section{Example}

In this section, simulations are conducted to demonstrate the performance of
the presented AUSMAMC law~(\ref{sliding mode control}) for rest-to-rest
attitude maneuvers of a rigid spacecraft. In addition, the existing
controller (11) in~\cite{zhu2011adaptive} and (23) in~\cite%
{tiwari2018spacecraft} are adopted for comparison.

\subsection{Simulation Settings}

\subsubsection{Spacecraft parameter values}

\label{system parameters} The inertia matrix of the rigid spacecraft is $J=%
\left[ 20\ 0\ 0.9;0\ 17\right.$ $\left. 0; 0.9\ 0\ 15\right] \mathrm{kg\cdot
m}^{2}$. The initial value of the attitude velocity $\boldsymbol{\omega}$
and quaternion $\boldsymbol{q}$ are given in TABLE~\ref{initialvalue}. The
disturbance is $\boldsymbol{d}=10^{-2}\times\left[ \sin \left( 0.05t\right)\
0.5\sin \left( 0.05t\right)\ -\cos \left( 0.05t\right) \right] ^{\mathrm{T}}$%
.
\begin{table}[thb]
\caption{Initial value of the signal of the rigid spacecraft system}
\label{initialvalue}\centering
{%
\scalebox{0.9}{\scriptsize\begin{tabular}{llll}
\Xhline{1pt} Notation & Unit & Meaning & Initial value \\ \hline
\vspace{0.1cm} $\boldsymbol{q}\left(0\right)\in\mathbb{R}^{4}$ & / & \makecell[{ll}]{ Attitude of
$\mathcal{F}_{\mathrm{b}}$ with respect to $\mathcal{F}_{\mathrm{I}}$} & $\left[
1\ 0\ 0\ 0\right] ^{\text{T}} $ \\
\vspace{0.1cm} $q_{0}\left(0\right) \in \mathbb{R}$ & / & Scalar part of $\boldsymbol{q}$
& $1$ \\
$\boldsymbol{q}_{\mathrm{v}}\left(0\right)\in\mathbb{R}^{3}$ & / & Vector part of $\boldsymbol{q}$ & $\left[0\ 0\ 0\right]^{\text{T}} $ \\
\vspace{0.1cm} $\boldsymbol{\omega}\left(0\right)\in \mathbb{R}^{3}$ & $\text{rad/s} $ & \makecell[{ll}]{Attitude of $\mathcal{F}_{\mathrm{b}}$ with respect to
$\mathcal{F}_{\mathrm{I}}$ } & $\left[ 0\ 0\ 0\right] ^{\text{T}}$ \\
\Xhline{1pt} &  &  &
\end{tabular}}}
\end{table}

\subsubsection{Controller parameter values}

\label{controller parameter values} The tuning parameters of the proposed
AUSMAMC law (\ref{sliding mode control}) are chosen by a trail-to-trail
selection. The parameters of the controller (11) \cite{zhu2011adaptive}, and
controller (23) \cite{tiwari2018spacecraft} are chosen the same as in \cite{zhu2011adaptive} and \cite{tiwari2018spacecraft}, respectively. The value of the parameters of the
above controllers are shown in TABLE~\ref{tuning parameter}. In addition, $%
\gamma_{2}\left(t\right)$ can be obtained from~(\ref{gamma2}).
\begin{table}[thb]
\caption{Control parameters chosen for numerical analysis}
\label{tuning parameter}\centering
{%
\scalebox{0.9}{\scriptsize
\begin{tabular}{lll}
\Xhline{1pt} Control schemes & control parameters &  \\ \hline
\vspace{0.1cm} \makecell[{cc}]{AUSMAMC (\ref{sliding mode control})} & \makecell[{ll}]{ $\lambda=2,\ \gamma_1=10,\ \varepsilon=0.5$} &  \\
\Xhline{0.5pt}
\vspace{0.1cm} \makecell[{cc}]{\makecell[{ll}]{Controller (11)\\~\protect\cite{zhu2011adaptive}} } & \makecell[{ll}]{$k=1,\ \tau=15I_{3},\ \sigma=0.001I_{3},$\\$\ p_{0}
=1,p_{1}=1,p_{2}=1,$ \\$\hat{c}\!\left(0\right)
\!=\!1,\ \hat{k}_{1}\!\left(0\right)\!=\!0.1,\ \hat{k}_{2}\!\left(0\right)\!=\!0.1$} &  \\
\Xhline{0.5pt}
\vspace{0.1cm} \makecell[{cc}]{\makecell[{ll}]{Controller (23)\\~\protect\cite{tiwari2018spacecraft}} } & \makecell[{ll}]{$v_{1}=5I_{3}, v_{2}=7I_{3}, \rho=0.001I_{3},$\\$k
=1.5,\eta=0.14, \theta=7,$ \\$\hat{K}_{1}\left(0\right)
=0,\ \hat{K}_{2}\left(0\right)=0$} &  \\
\Xhline{1pt} &  &
\end{tabular}}}
\end{table}

\subsubsection{Control goal}

\label{control goal} The control goal is to perform two rest-to-rest
attitude maneuvers for the rigid spacecraft with system parameters given in
Section~\ref{system parameters}.
%using the proposed controller AUSMAMC law~(\ref{sliding mode control}).
%The control objective in Section \ref{controobjective} can be satisfied. %The pointing precision is within $0.002\text{deg}$ error, and during the entire attitude maneuver motion, the unwinding phenomenon is avoided.
Two different scenarios of desired attitude value are given in the following.

Scenario A. The initial values of the desired quaternion and angular
velocity are $\boldsymbol{q}_{\mathrm{d}}=\left[ 0.8832\ 0.3\ -0.2\ -0.3%
\right] ^{\mathrm{T}},$ and $\boldsymbol{\omega}_{\mathrm{d}} =\left[ 0\ 0\ 0%
\right] ^{\mathrm{T}}$rad/s, respectively.

Scenario B. The initial values of the desired quaternion and angular
velocity are $\boldsymbol{q}_{\mathrm{d}}=\left[ -0.6403\ -0.5\ -0.3\ 0.5%
\right] ^{\mathrm{T}},$ and $\boldsymbol{\omega}_{\mathrm{d}} =\left[ 0\ 0\ 0%
\right] ^{\mathrm{T}}$rad/s, respectively.

In Scenario A, $\boldsymbol{q}_{\mathrm{d}0}>0$, thus $\boldsymbol{q}_{%
\mathrm{e}}=\left[ 1\ 0\ 0\ 0\right] ^{\mathrm{T}}$ is the nearest
equilibrium. In addition, according to the first equation of (\ref{unit
attitude error}), the spacecraft needs to rotate $55.93^{\circ}$ to reach
the equilibrium point. In Scenario B, $\boldsymbol{q}_{\mathrm{d}%
0}\left(0\right)<0$, and the spacecraft needs to tilt $259.62^{\circ}$ if
only the equilibrium $\boldsymbol{q}_{\mathrm{e}}=\left[ 1\ 0\ 0\ 0 \right]
^{\mathrm{T}}$ is considered. However, the spacecraft only need to rotate $%
100.38^{\circ}$ if $\boldsymbol{q}_{\mathrm{e}}=\left[ -1\ 0\ 0\ 0\right] ^{%
\mathrm{T}}$ is also considered as an equilibrium.
%Further, in order to verify the superiority of the presented AUSMAMC law (\ref{sliding mode control}), the controller (11) \cite{zhu2011adaptive} and controller (23)~\cite{tiwari2018spacecraft} are adopted for comparison in Scenarios A and B.

\subsection{Simulation results}

\subsubsection{Simulation results for Scenario A}

The controller (11)~\cite{zhu2011adaptive} and proposed AUSMAMC law (\ref%
{sliding mode control}) are adopted to do simulations for Scenario A. The
simulation results are shown in Fig. \ref{fig_2}, where the controller A
represents the controller (11)~\cite{zhu2011adaptive}.

The response of error quaternions $q_{\mathrm{e}i}, i=1,2,3,4$ and angular
velocity error $\omega_{\mathrm{e}i}, i=1,2,3$ are shown in Fig.~\ref{fig_A1}
and Fig.~\ref{fig_C1}, respectively. It can be seen from Fig. \ref{fig_A1}
and Fig. \ref{fig_C1} that the attitude errors of system~(\ref{system model}%
) converge to $0$ in about $4\mathrm{s}$ by adopting the proposed AUSMAMC
law (\ref{sliding mode control}), while the Controller A needs longer time.
In addition, it can be easily obtained from these two figures that the
steady attitude errors of the developed control law AUSMAMC law (\ref%
{sliding mode control}) are smaller than that of the Controller A. The
spacecraft attitude responses using Euler angles $\phi, \theta, \psi$ ($%
\phi, \theta, \psi$ are the roll, pitch, and yaw angles, respectively) are
shown in Fig. \ref{fig_E1}, which indicates that the attitude maneuver
problem can be effectively settled by the controller AUSMAMC law (\ref%
{sliding mode control}) and Controller A. %The attitude rotation of
%the rigid spacecraft is shown in Fig. \ref{fig_E1}.
%It can be seen from Fig. % \ref{fig_E1} that the rotation angle of the rigid spacecraft is $55.93^{\circ}$.
The time evolution of control torques $u_{i},i=1,2,3$ are shown in Fig. \ref%
{fig_D1}. The control torque of the proposed AUSMAMC law (\ref{sliding mode
control}) is smaller than that of Controller A.
%It can also be observed that the proposed control law possesses better performance than the controller (11)~\cite{%zhu2011adaptive}.

The AUSMAMC controller is able to obtain higher pointing accuracy and better
stability in a shorter time.
%In other words, both the rapidity and pointing accuracy of the attitude system are improved by using the presented controller AUSMAMC law (\ref{sliding mode control}).
%In other words, the control objectives 1) and 2) in Section \ref{controobjective} are satisfied by using
%the presented controller AUSMAMC law (\ref{sliding mode control}). %In other words, both the rapidity and pointing accuracy of the attitude system are improved by using the presented controller AUSMAMC law (\ref{sliding mode control}).
\begin{figure*}[t]
\centering
\subfigure[Time response of quaternion $q$]{
		\label{fig_A1}
		\includegraphics[width=2.5in]{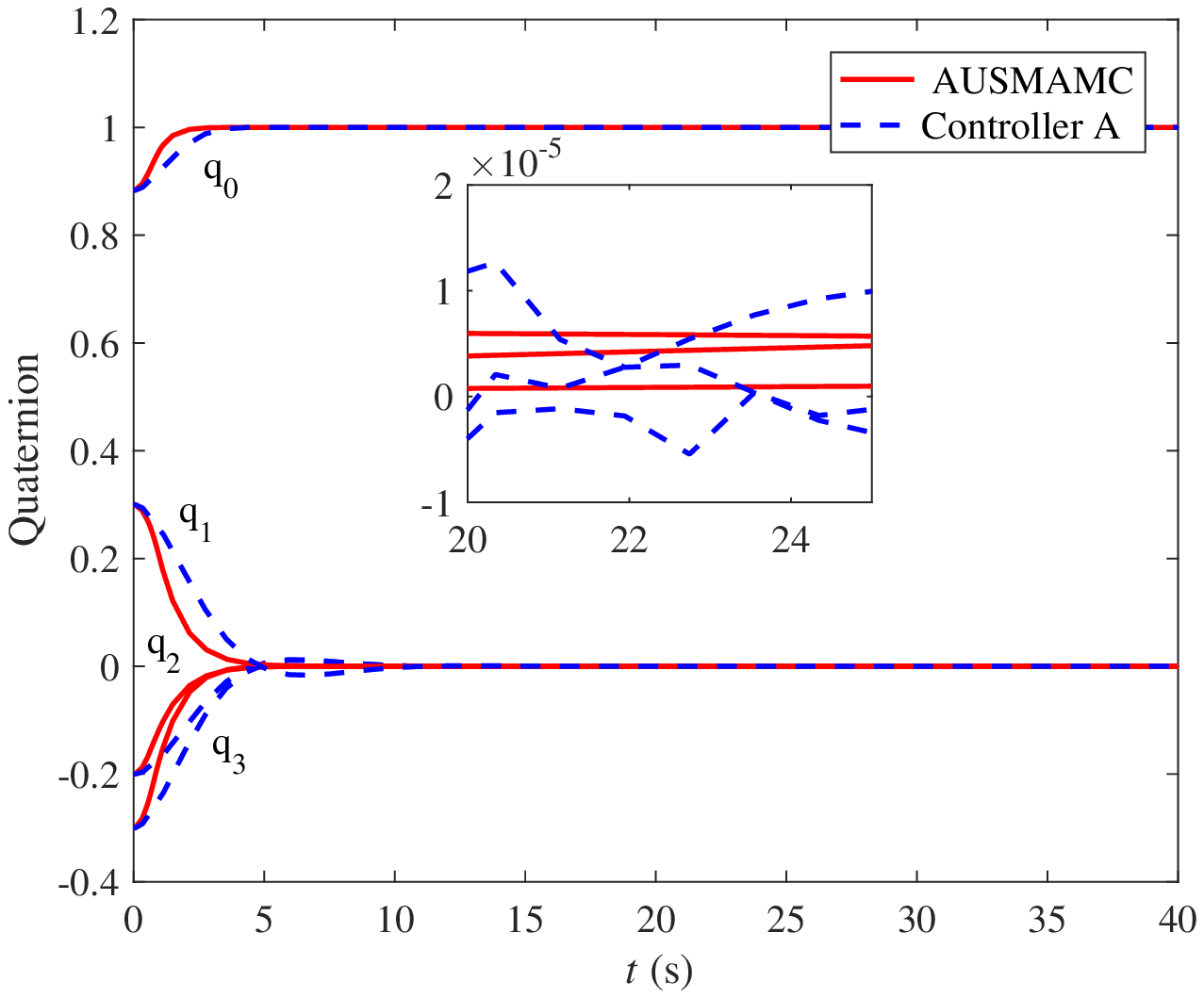}} %1.85
\hspace{1cm}
\subfigure[Time response of the angular velocity]{
		\label{fig_C1}
		\includegraphics[width=2.5in]{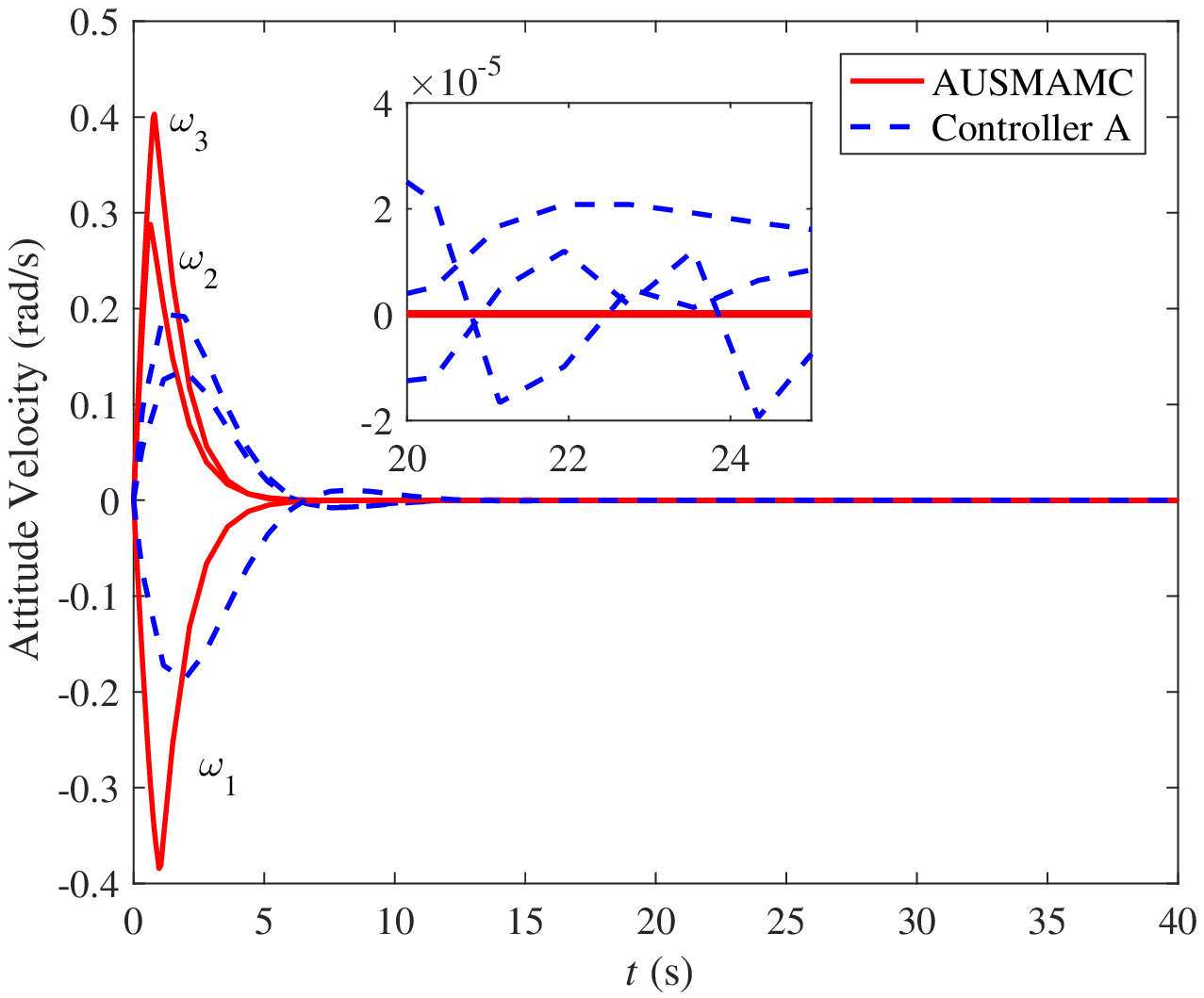}} %1.85
\subfigure[Evolution of the Euler angles $\phi,\theta, \psi $ for the Scenario A]{
		\label{fig_E1}
		\includegraphics[width=2.5in]{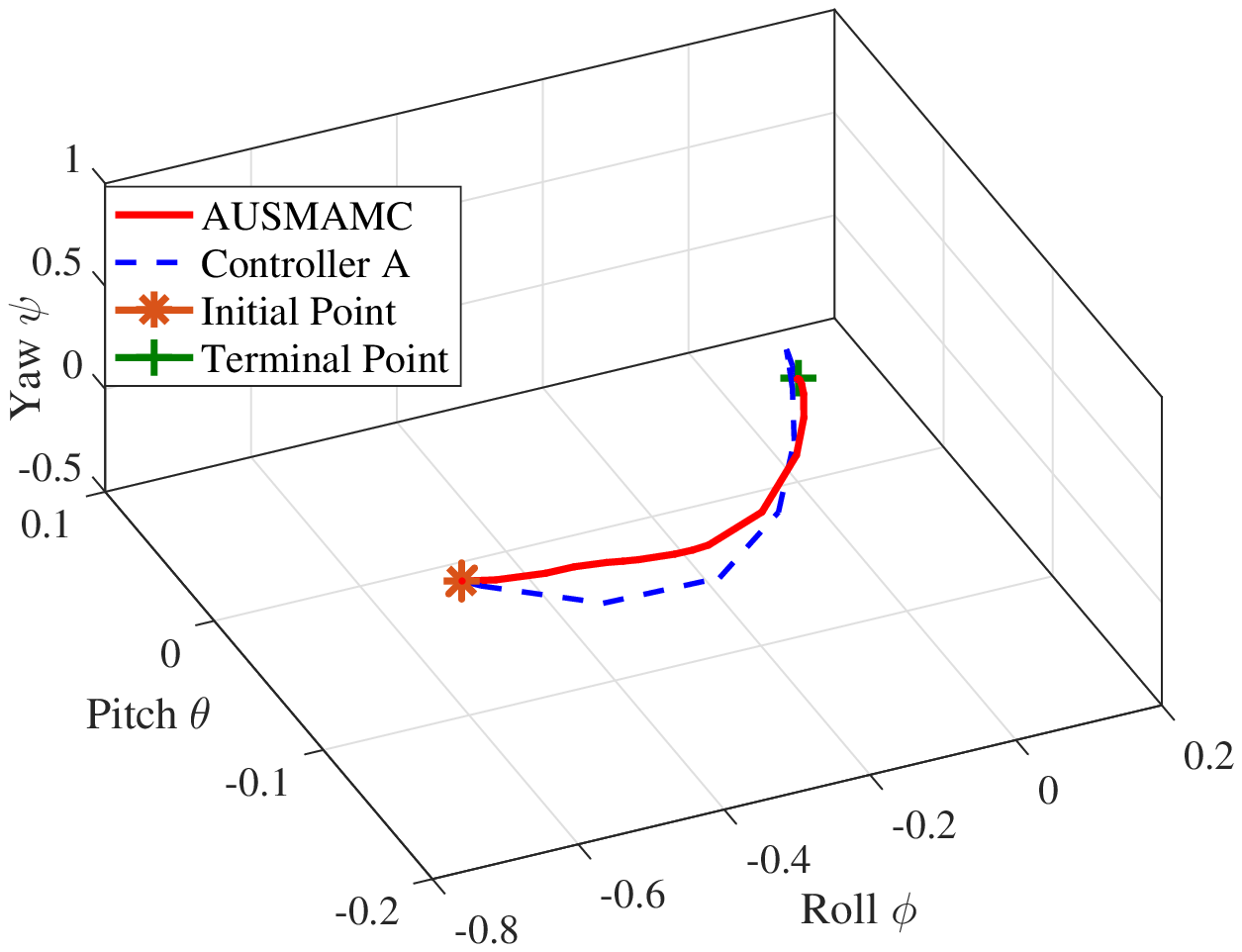}} %1.85
\hspace{1cm}
\subfigure[Time response of the control torques]{
		\label{fig_D1}
		\includegraphics[width=2.5in]{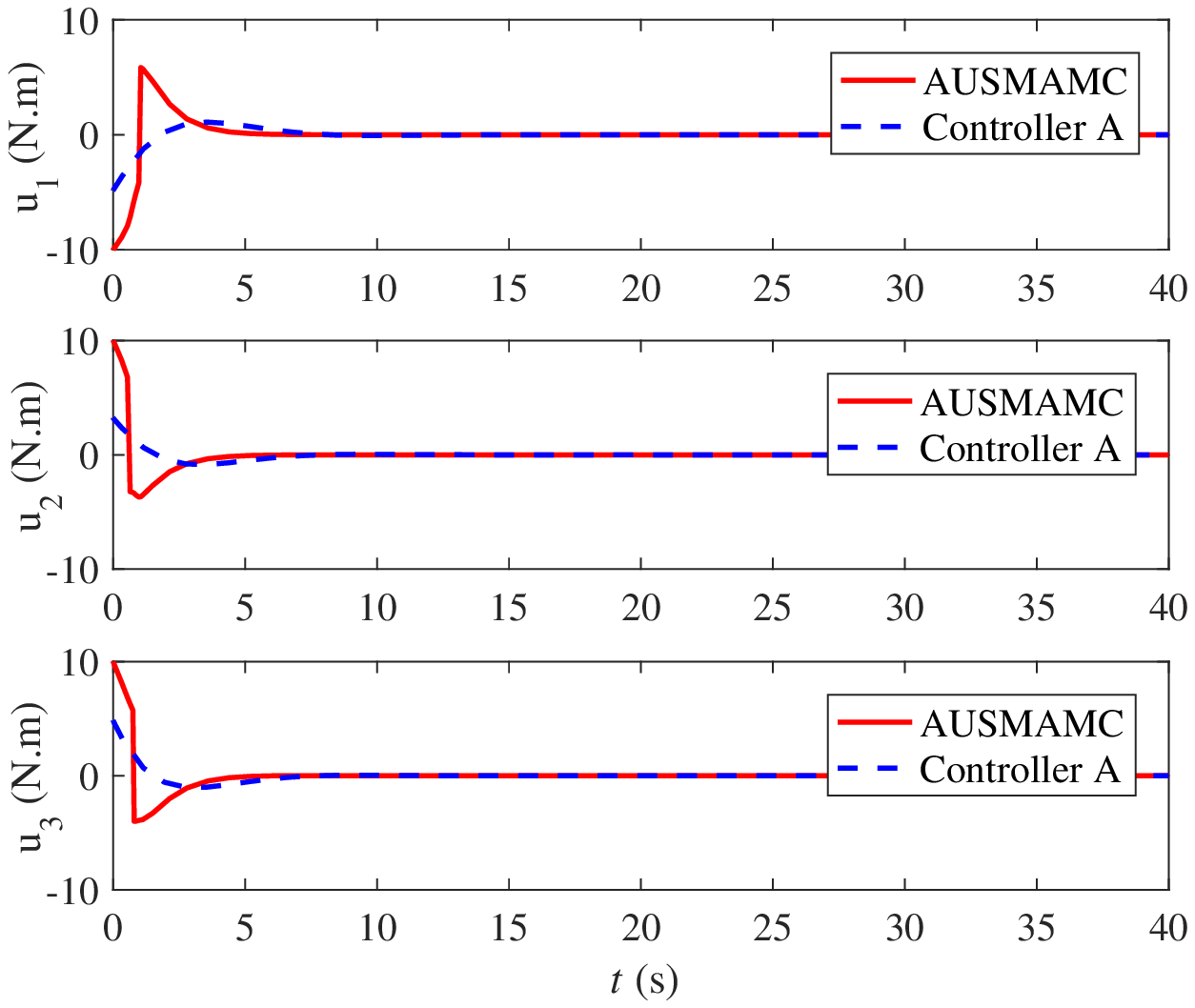}}%1.85
\caption{Comparison results of AUSMAMC law (\protect\ref{sliding mode
control}) and controller (11)~\protect\cite{zhu2011adaptive} for Scenario A }
\label{fig_2}
\end{figure*}

\subsubsection{Simulation results for Scenario B}

The controller (11)~\cite{zhu2011adaptive}, controller (23)~\cite%
{tiwari2018spacecraft}, and the proposed AUSMAMC law (\ref{sliding mode
control}) are adopted to do simulations for Scenario B. The simulation
results are summarized in Fig. \ref{fig_3}, where the controller A is the
controller (11)~\cite{zhu2011adaptive}, and controller B is the controller
(23)~\cite{tiwari2018spacecraft}.

The response of error quaternions $q_{\mathrm{e}i}, i=1,2,3,4$ are shown in
Fig. \ref{fig_A3}, which indicates that $\boldsymbol{q}_{\mathrm{e}}$
converges to the nearest equilibrium $[-1\ 0\ 0\ 0]$ in about $5\mathrm{s}$
by adopting the presented controller AUSMAMC (\ref{sliding mode control})
and Controller B. However, $\boldsymbol{q}_{\mathrm{e}}$ converges to $[1\
0\ 0\ 0]$ in about $14\mathrm{s}$ by adopting the Controller A.
%In other words, the rigid spacecraft only needs to rotate $100.38^{\circ}$ to reach the desired attitude by the controller AUSMAMC (\ref{sliding mode control}) and (23)\cite{compare1}, while $259.62^{\circ}$
Thus, it can be obtained that the presented AUSMAMC law (\ref{sliding mode
control}) in this paper and (23)~\cite{tiwari2018spacecraft} avoids
unwinding phenomenon successfully, but the Controller A suffers unwinding
problem. The behaviour of angular velocity error $\omega_{\mathrm{e}i},
i=1,2,3$ is shown in Fig. \ref{fig_C3}. It can be observed from Fig. \ref%
{fig_C3} that the attitude velocity of the rigid spacecraft (\ref{system
model}) converges to $0$ in about $5\mathrm{s}$ by using the proposed
AUSMAMC law (\ref{sliding mode control}) and Controller B, while the
Controller A needs longer time. In addition, it can be easily obtained from
these two figures that the steady attitude errors of the developed control
law AUSMAMC law~(\ref{sliding mode control}) are smaller than that of the
Controller A and Controller B. The spacecraft attitude responses using Euler
angles $\phi, \theta, \psi$ ($\phi, \theta, \psi$ are the roll, pitch, and
yaw angles, respectively) are shown in Fig. \ref{fig_B3}. The maneuver angle
of the AUSMAMC law~(\ref{sliding mode control}) and Controller B is smaller
than that of Controller A. The control torques $u_{i},i=1,2,3$ are shown in
Fig. \ref{fig_D3}, which indicates that the attitude maneuver is
effectively settled by the controller AUSMAMC law (\ref{sliding mode control}%
), the Controllers A and B. It can also be observed that the control torque
of the proposed control law is less than that of the Controllers A and B.
\begin{figure*}[t]
\centering
\subfigure[Time response of quaternion $q$]{
		\label{fig_A3}
		\includegraphics[width=2.5in]{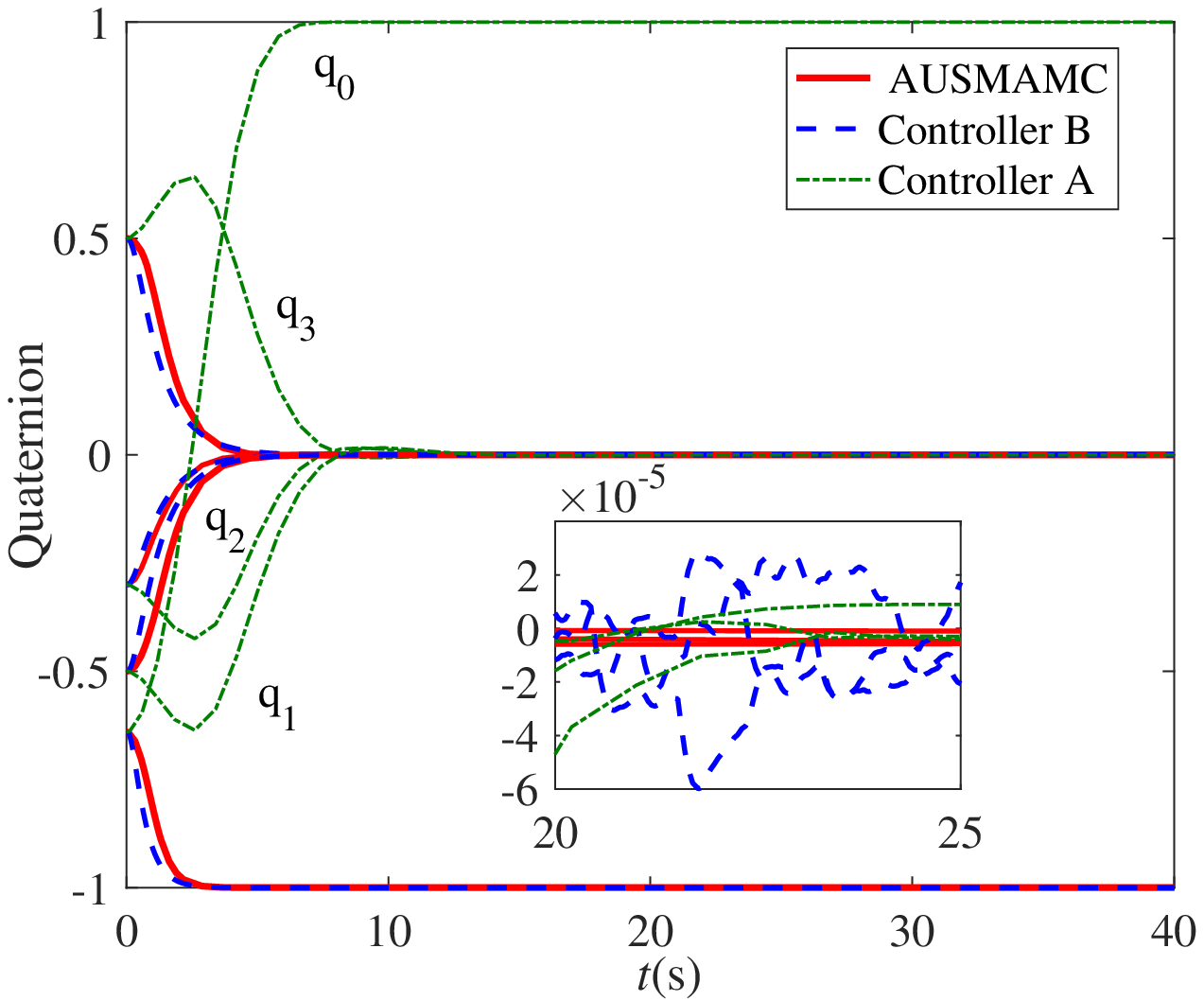}} %1.85
\hspace{1cm}
\subfigure[Time response of the angular velocity]{
		\label{fig_C3}
		\includegraphics[width=2.5in]{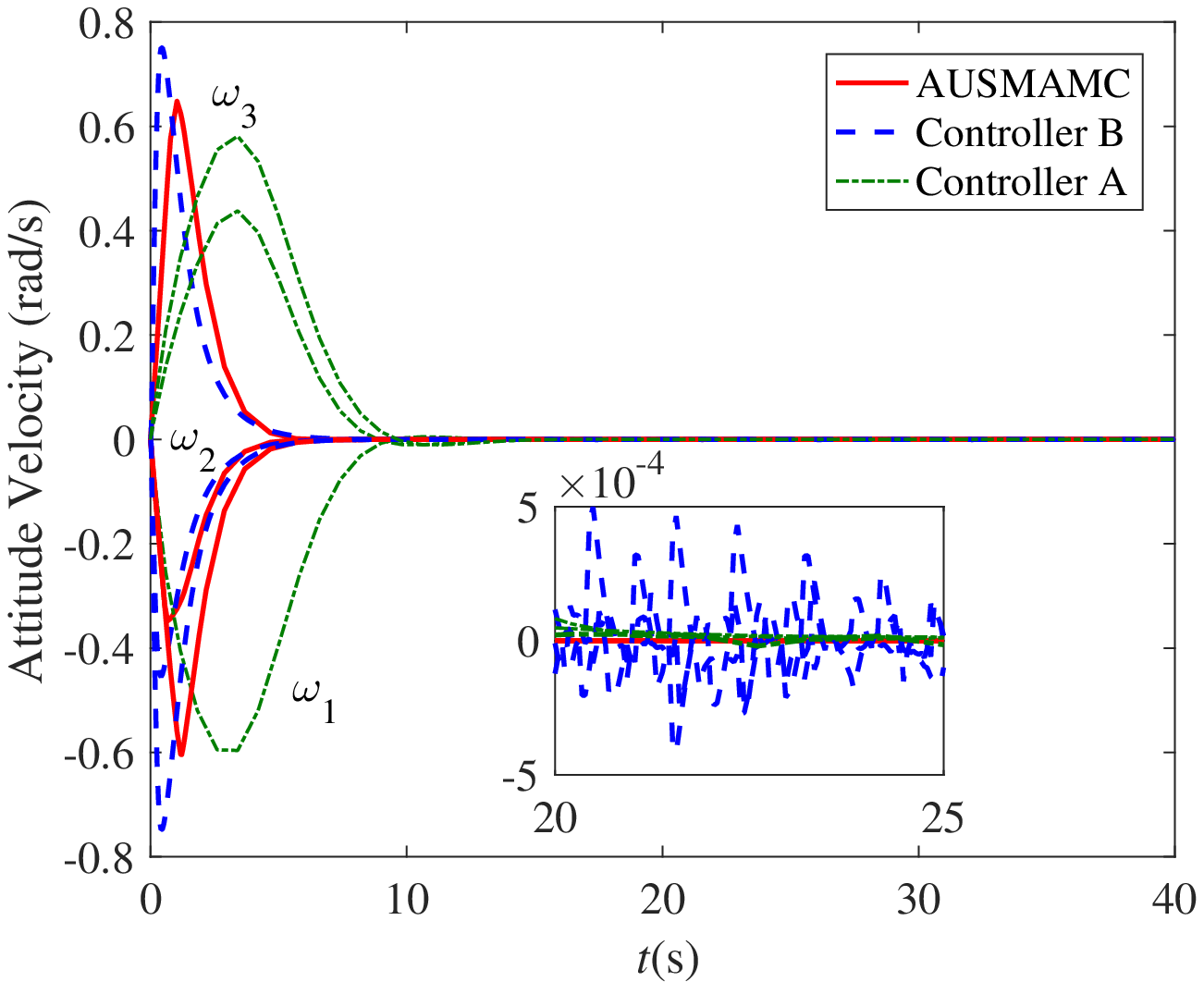}} %1.85
\subfigure[Evolution of the Euler angles $\phi,\theta, \psi $ for the Scenario B]{
		\label{fig_B3}
		        \includegraphics[width=2.5in]{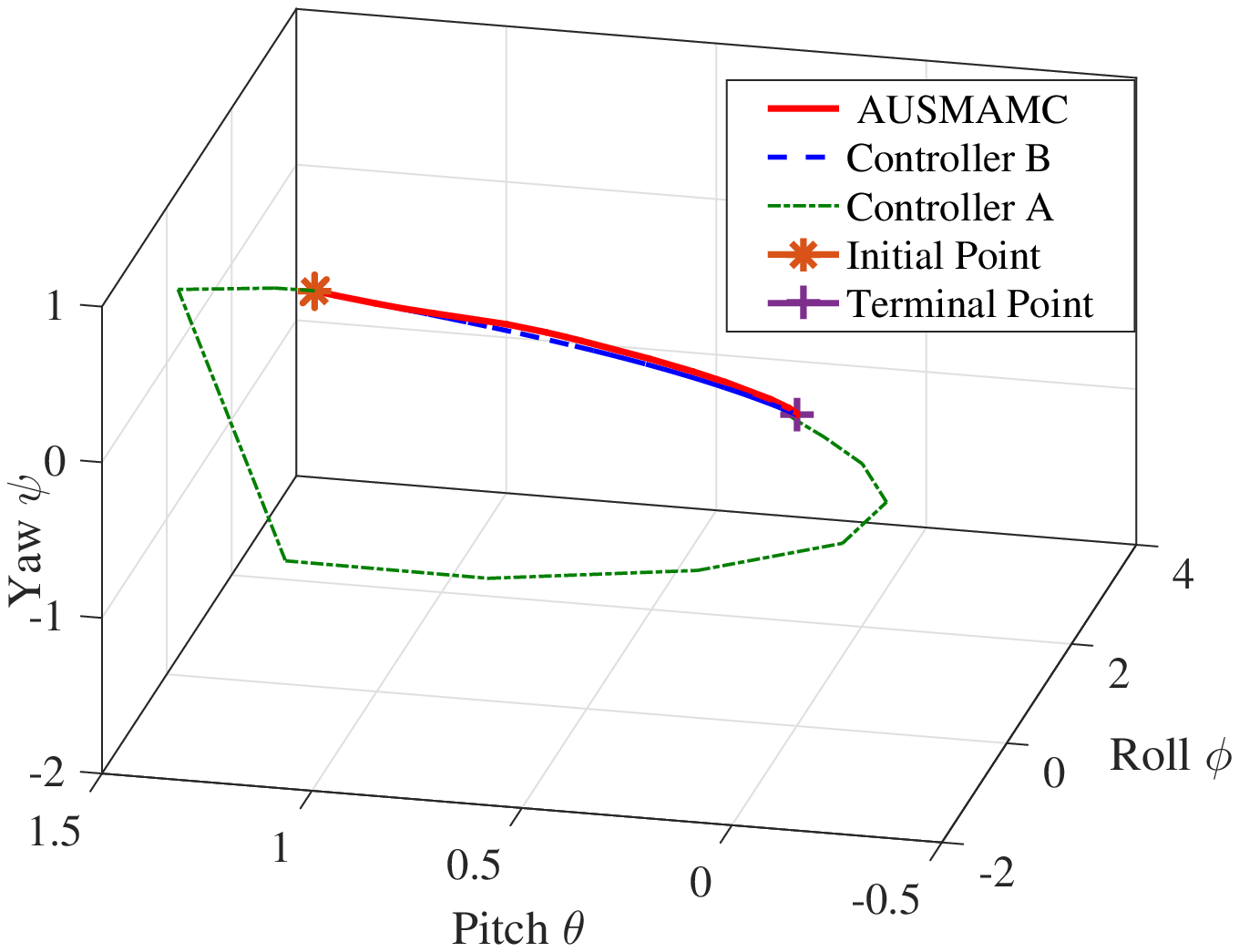}} %1.15
\hspace{1cm}
\subfigure[Time response of the control torques]{\label{fig_D3}
		\includegraphics[width=2.5in]{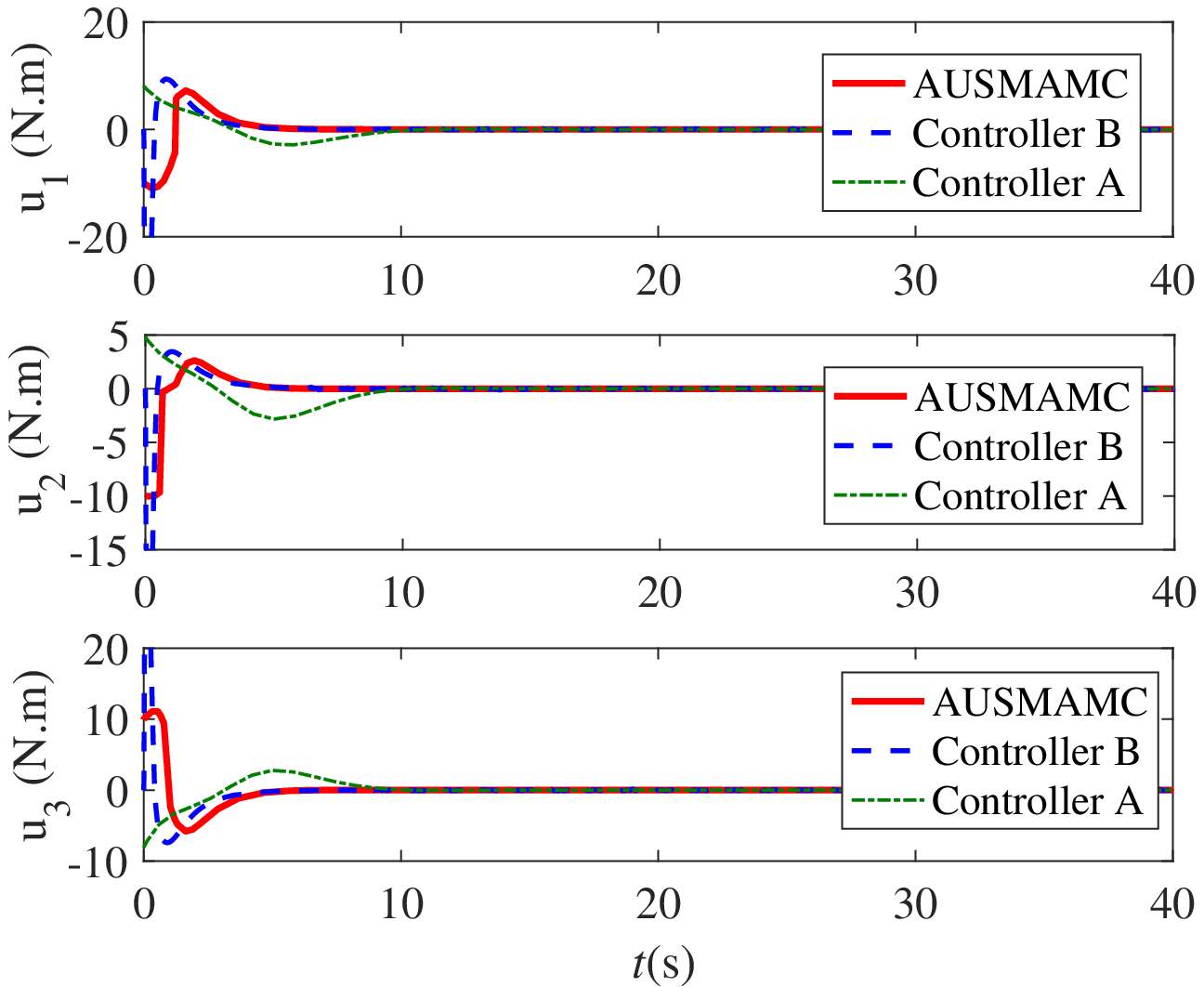}} %1.85
\caption{Comparison results of AUSMAMC law (\protect\ref{sliding mode
control}) and controllers (11)~\protect\cite{zhu2011adaptive} and (23)~%
\protect\cite{tiwari2018spacecraft} for Scenario B}
\label{fig_3}
\end{figure*}

In conclusion, the proposed AUSMAMC controller~(\ref{sliding mode control})
satisfies the control objective described in Section \ref{controobjective},
and it achieves higher pointing accuracy and better stability in a shorter
time compared with the controller (11) \cite{zhu2011adaptive}, and
controller (23) \cite{tiwari2018spacecraft}.%
%In other words, both the rapidity and pointing accuracy of the attitude system are improved by using the presented controller AUSMAMC law (\ref{sliding mode control}). Furthermore, the developed AUSMAMC law (\ref{sliding mode control}) is able to avoid unwinding phenomenon.

\section{Conclusion}

In this paper, an anti-unwinding attitude maneuver control law is presented
for rigid spacecraft. By constructing a new switching
surface, which contains two equilibriums, the unwinding problem is settled
when the system states are on the switching surface. Moreover, by designing
a sliding mode control law with a dynamic parameter, the anti-unwinding
performance is guaranteed before the system states reach the switching
surface.
%The anti-unwinding sliding mode attitude control law has the robustness to the inertia uncertainty and disturbance.
Further, the switching function, the attitude velocity error, and the vector
part of error quaternion converge to zero under the designed anti-unwinding
sliding mode attitude maneuver control law. Finally, a numerical simulation
is conducted to demonstrate the effectiveness of the developed control law.
The simulation results show that the unwinding phenomenon is avoided by
adopting the designed switching surface and controller.

%\nocite{*}

\bibliographystyle{ieeetr}
\bibliography{myreference}

\begin{thebibliography}{10}

\bibitem{li2016robust}
Y.~Li, Z.~Sun, and D.~Ye, ``Robust linear pid controller for satellite attitude
  stabilisation and attitude tracking control,'' {\em International Journal of
  Space Science and Engineering}, vol.~4, no.~1, pp.~64--75, 2016.

\bibitem{jin2018lpv}
R.~Jin, X.~Chen, Y.~Geng, and Z.~Hou, ``Lpv gain-scheduled attitude control for
  satellite with time-varying inertia,'' {\em Aerospace Science and
  Technology}, vol.~80, pp.~424--432, 2018.

\bibitem{chen2000adaptive}
B.~S. Chen, C.~S. Wu, and Y.~W. Jan, ``Adaptive fuzzy mixed ${H}_2/{H}_\infty$
  attitude control of spacecraft,'' {\em IEEE Transactions on Aerospace and
  Electronic Systems}, vol.~36, no.~4, pp.~1343--1359, 2000.

\bibitem{guo2019velocity}
Y.~Guo, B.~Huang, J.~Guo, A.~Li, and C.~Wang, ``Velocity-free sliding mode
  control for spacecraft with input saturation,'' {\em Acta Astronautica},
  vol.~154, pp.~1--8, 2019.

\bibitem{gao2012robust}
X.~Gao, K.~L. Teo, and G.~Duan, ``Robust ${H_\infty}$ control of spacecraft
  rendezvous on elliptical orbit,'' {\em Journal of the Franklin Institute},
  vol.~349, no.~8, pp.~2515--2529, 2012.

\bibitem{utkin1977variable}
V.~Utkin, ``Variable structure systems with sliding modes,'' {\em IEEE
  Transactions on Automatic control}, vol.~22, no.~2, pp.~212--222, 1977.

\bibitem{dodds1991sliding}
S.~Dodds and A.~Walker, ``Sliding-mode control system for the three-axis
  attitude control of rigid-body spacecraft with unknown dynamics parameters,''
  {\em International Journal of control}, vol.~54, no.~4, pp.~737--761, 1991.

\bibitem{bang2005flexible}
H.~Bang, C.~K. Ha, and J.~H. Kim, ``Flexible spacecraft attitude maneuver by
  application of sliding mode control,'' {\em Acta Astronautica}, vol.~57,
  no.~11, pp.~841--850, 2005.

\bibitem{hu2006control}
Q.~Hu and G.~Ma, ``Control of three-axis stabilized flexible spacecraft using
  variable structure strategies subject to input nonlinearities,'' {\em Journal
  of Vibration and Control}, vol.~12, no.~6, pp.~659--681, 2006.

\bibitem{hu2006adaptive}
Q.~Hu and G.~Ma, ``Adaptive variable structure maneuvering control and
  vibration reduction of three-axis stabilized flexible spacecraft,'' {\em
  European Journal of Control}, vol.~12, no.~6, pp.~654--668, 2006.

\bibitem{chen1993sliding}
Y.~P. Chen and S.~C. Lo, ``Sliding-mode controller design for spacecraft
  attitude tracking maneuvers,'' {\em IEEE Transactions on Aerospace and
  Electronic Systems}, vol.~29, no.~4, pp.~1328--1333, 1993.

\bibitem{zong2010higher}
Q.~Zong, Z.~Zhao, and J.~Zhang, ``Higher order sliding mode control with
  self-tuning law based on integral sliding mode,'' {\em IET Control Theory \&
  Applications}, vol.~4, no.~7, pp.~1282--1289, 2010.

\bibitem{tiwari2016attitude}
P.~M. Tiwari, S.~Janardhanan, and M.~Nabi, ``Attitude control using higher
  order sliding mode,'' {\em Aerospace Science and Technology}, vol.~54,
  pp.~108--113, 2016.

\bibitem{li2006global}
J.~Li and C.~Qian, ``Global finite-time stabilization by dynamic output
  feedback for a class of continuous nonlinear systems,'' {\em IEEE
  Transactions on Automatic control}, vol.~51, no.~5, pp.~879--884, 2006.

\bibitem{wu2011quaternion}
S.~Wu, G.~Radice, Y.~Gao, and Z.~Sun, ``Quaternion-based finite time control
  for spacecraft attitude tracking,'' {\em Acta Astronautica}, vol.~69,
  no.~1-2, pp.~48--58, 2011.

\bibitem{wu2018adaptive}
A.~Wu, R.~Dong, Y.~Zhang, and L.~He, ``Adaptive sliding mode control laws for
  attitude stabilization of flexible spacecraft with inertia uncertainty,''
  {\em IEEE Access}, vol.~7, pp.~7159--7175, 2018.

\bibitem{ji2018vibration}
N.~Ji and J.~Liu, ``Vibration control for a flexible satellite with input
  constraint based on nussbaum function via backstepping method,'' {\em
  Aerospace Science and Technology}, vol.~77, pp.~563--572, 2018.

\bibitem{guo2014adaptive}
Y.~Guo and S.~Song, ``Adaptive finite-time backstepping control for attitude
  tracking of spacecraft based on rotation matrix,'' {\em Chinese Journal of
  Aeronautics}, vol.~27, no.~2, pp.~375--382, 2014.

\bibitem{hu2015spacecraft}
Q.~Hu, L.~Li, and M.~I. Friswell, ``Spacecraft anti-unwinding attitude control
  with actuator nonlinearities and velocity limit,'' {\em Journal of Guidance,
  Control, and Dynamics}, vol.~38, no.~10, pp.~2042--2050, 2015.

\bibitem{tiwari2018spacecraft}
P.~M. Tiwari, S.~Janardhanan, and M.~Nabi, ``Spacecraft anti-unwinding attitude
  control using second-order sliding mode,'' {\em Asian Journal of Control},
  vol.~20, no.~1, pp.~455--468, 2018.

\bibitem{kristiansen2005satellite}
R.~Kristiansen and P.~J. Nicklasson, ``Satellite attitude control by
  quaternion-based backstepping,'' in {\em Proceedings of the 2005, American
  Control Conference, 2005.}, pp.~907--912, IEEE, 2005.

\bibitem{di2003output}
S.~Di~Gennaro, ``Output stabilization of flexible spacecraft with active
  vibration suppression,'' {\em IEEE Transactions on Aerospace and Electronic
  systems}, vol.~39, no.~3, pp.~747--759, 2003.

\bibitem{Lemmastability}
S.~P. Bhat and D.~S. Bernstein, ``Finite-time stability of continuous
  autonomous systems,'' {\em SIAM Journal on Control and Optimization},
  vol.~38, no.~3, pp.~651--766, 2000.

\bibitem{zhu2011adaptive}
Z.~Zhu, Y.~Xia, and M.~Fu, ``Adaptive sliding mode control for attitude
  stabilization with actuator saturation,'' {\em IEEE Transactions on
  Industrial Electronics}, vol.~58, no.~10, pp.~4898--4907, 2011.

\end{thebibliography}

\end{document}